\documentclass[11pt,letterpaper,oneside, usenames,dvipsnames]{article}

\usepackage{amssymb}
\usepackage{amsmath}
\usepackage{mathrsfs}

\usepackage{parskip}

\usepackage{amsthm}
\usepackage[latin1]{inputenc}
\usepackage{xcolor}

\usepackage{bbm}

\usepackage{textcomp}
\usepackage{multicol,boxedminipage}

\usepackage{verbatim} 
\usepackage[toc,page]{appendix}

\usepackage{float}

\usepackage[authoryear]{natbib}
\usepackage{marvosym}

\usepackage{tikz}

\usepackage[linesnumbered,ruled,vlined]{algorithm2e}

\usepackage{mathtools}

\usepackage{float}
\usepackage[hypertexnames=false]{hyperref}
\usepackage{subcaption}  

\let\oldproof\proof
\renewcommand{\proof}{\color{black!60!blue}\oldproof}

\newtheorem{thm}{Theorem}

\newtheorem{lemma}[thm]{Lemma}
\newtheorem{cor}[thm]{Corollary}

\theoremstyle{remark}
\newtheorem{rem}[thm]{Remark}

\theoremstyle{definition}

\newcommand{\Px}{\mathbb{P}}
\newcommand{\Ex}{\mathbb{E}}
\newcommand{\R}{\mathbb{R}}

\newcommand{\Dt}{\Delta t}
\newcommand{\dt}{\Delta t}
\newcommand{\Dl}[1]{\Delta \lambda^\infty_{#1}(i,\rho)}
\newcommand{\linf}{\lambda^\infty(i,\rho)}

\newcommand{\XX}{\mathcal{X}}
\newcommand{\YY}{\mathcal{Y}}
\newcommand{\KK}{\mathcal{K}}
\newcommand{\Xt}{\overline{\mathbf{X}}_t}
\newcommand{\Yt}{\overline{\mathbf{Y}}_t}

\newcommand{\sbt}{\,\begin{picture}(-1,1)(-1,-2)\circle*{2}\end{picture}\ }  

\definecolor{gold}{rgb}{0.64,0.54,0.29}

\newcommand{\sk}{\smallskip}
\newcommand{\mk}{\medskip}
\newcommand{\bk}{\bigskip}

\newcounter{defcounter}
\newenvironment{assumptions}{%
\addtocounter{equation}{-1}
\refstepcounter{defcounter}

\begin{equation}}
{\end{equation}}

\usepackage{geometry}
\usepackage{float}

\title{A Stochastic Model for the Early Stages of Highly Contagious Epidemics by using a State-Dependent Point Process}
\author{Jonathan Ch\'{a}vez-Casillas\thanks{Department of Mathematics and Applied Mathematical Sciences, University of Rhode Island, USA ({\tt jchavezc@uri.edu}).}
  }

\begin{document}

\maketitle

\begin{abstract}
The recent COVID-19 pandemic has shown that when the reproduction number is high and there are no proper measurements in place, the number of infected people can increase dramatically in a short time, producing a phenomenon that many stochastic SIR-like models cannot describe: overdispersion of the number of infected people (i.e., the variance of the number of infected people during any interval is very high compared to the average). To address this issue, in this paper we explore the possibility of modeling the total number of infections as a state dependent self-exciting point process. In this way, infections are not independent among themselves, but any infection  will increase the likelihood of a new infection while also the number of currently infected and recovered individuals are included into determining the likelihood of new infections, Since long term simulation is extremely computationally intensive, exact expressions for the moments of the processes determining the number of infected and recovered individuals are computed, while also simulation algorithms for these state-dependent processes are provided.
\end{abstract}

\bk


\mk

\section{Introduction}

As the, now seemingly endemic, disease of COVID-19 is being incorporated into our daily lives, we have been reminded that epidemics are a phenomena that can expand quickly and have devastating economic and social consequences. Highly contagious diseases can produce huge damage, and so the development of accurate and effective models for epidemics is important. However, as noted in  \cite{adiga2020data}, different kinds of models are appropriate
at different stages, and for addressing different kinds of questions. For example, some statistical methods based on machine learning techniques are very useful in predicting the behaviour on the short-term. However, they are not very effective for the long term predictions nor for describing the evolving circumstances. Simple compartmental-type models, and their extensions, that is, structured meta-population models, are useful for
several population-level questions. However, once the outbreak has spread, and complex societal interactions are at play, stochastic agent-based models could be seen as a more robust and effective tool, since they allow for a more systematic representation of complex social interactions, individual
and collective behavioral adaptation, and public policies. With this idea of trying to be more precise in describing different stages of a pandemic the current paper tries to shed some light into how to incorporate certain salient features observed in the data to model future outbreaks of epidemics with a high control reproduction number.

Stochastic modeling of epidemics becomes relevant when the environmental or demographic variability -such as transmission, recovery, births, deaths, or environmental impacts- are too large to account for in a deterministic model (see \cite{allen2017primer}).

The ongoing pandemic has returned epidemic modeling at the forefront of worldwide public policy-making \cite{bertozzi2020}. Many scientist from different fields, particularly medicine, biology, mathematics, physics and chemistry (for a few, non-exhaustive, list of examples, see \cite{cheng2020regional},\cite{chiang2022hawkes},\cite{dos2021machine},\cite{franco2020feedback},  \cite{hazarika2020modelling}, \cite{liu2022model}, \cite{van2017reproduction})

Due to their versatility and capabilities to model properties that regular stochastic SIR models can't, point process models have been recently explored as a viable option. Models using these processes are  data  driven  and  are flexible enough to allow  for  parametric  or  nonparametric estimation of the reproduction number and the time scale at which the contagions occur (see   \cite{bertozzi2020}. Furthermore, at some level, they can also be viewed as stochastic versions of popular compartmental models used in epidemiology. Since these processes have just been recently introduced, before stating the model assumptions and results, in remainder of this introduction, an introduction to point processes is provided as well as how they have recently been used to model the novel COVID-19 epidemics.

\subsection{An introduction to Point processes} \label{SubSec:intro}

Point processes can be understood roughly as a random set of points in a space $\XX$. The topology and properties of the space can be very general, making them a very versatile tool for modeling different phenomena. In this setting, these points, may represent anything, but the most common usage of point processes are occurrences of events in time, location of objects in the space or a mixture of the previous two. Some of the most recognized and applied types of point processes are:

\begin{itemize}
	\item  \emph{Spatial Point process:} Here, the space $\XX$ represents some fixed ``geographical location'', generally, $\XX\subseteq\R^n$. In this instance the number of events occurring and their location  are random.
 \item \emph{Temporal point process:} In this case, the space $\XX$ is a totally ordered set. Generally, $\XX=[0,T]\subset\R$ and each point represents the times at which certain events occurred. In these kind of processes the number of events and the times at which those events occurred are random.
  \item \emph{Spatio-Temporal point process:} Here, the space is a combination of the previous two cases and each point represents a place and time for an event and usually $\XX\subseteq\R_+\times\R^n$.
\end{itemize}

\sk

It also important to mention that the previous cases are sometimes called ``unmarked point processes''. However, sometimes, we want to measure or model an additional characteristic of the point (other than time or location). In this instance, we can attach a ``mark'' to each random point to classify them as belonging to a certain class. This flexibility comes almost at no cost, since ``marked point processes'' are merely a generalization where each point belongs to a space $\YY=\XX\times\KK$, where $\XX$ is one of the spaces described above and $\KK$ is the so-called mark space.

\sk

Generally speaking, when talking about temporal point processes there are mainly two -not disjoint- ways to describe them. The first one was born in the neurophysiologist community where point processes are seen as random variables in a complicated functional space and are characterized and analyzed via their moments (such as the  \textit{average, variance, skewness, kurtosis, etc.}). The advantage of this point of view is that these moments can be associated to statistical quantities which can then be estimated, fitted and calibrated using data. The second way to characterize point processes is through its \textit{stochastic intensity}, which provides a summary, at any given time $t$, of the likelihood that some new future event arrives in the time interval $[t,t+h)$ given its history, i.e., given the times of all past events up until time $t$. This notion of intensity has been used extensively because of its analytical tractability. The resulting family of processes that possess an intensity function is called (stochastic) intensity-based point processes, and actually, contains almost all the point processes which have some practical interest. Exhaustive treatments containing all the important probabilistic and statistical features of Point processes can be found in, e.g.  \cite{verejonesv1,verejonesv2,bremaud2020point}.

A realization of a point process over $[0, \infty)$ is a sequence $\{T_0\}_{n\geq1}$ in $[0, \infty]$ such that

\[
T_0=0,\qquad\qquad T_n<\infty\qquad\text{and}\qquad T_n<T_{n+1}.
\]

\noindent and the point process is considered to be non-explosive if $\lim\limits_{n\to\infty} T_n=\infty$. Further, for each realization of the point process, there is a counting function, or counting process, $N_t$ defined as
\[
N_t=\left\{\begin{array}{rcl} n&\text{ if }& t\in[T_n,T_{n+1}),\ n\geq0\\ \infty &\text{ if }&t\geq \lim_{n\to\infty} T_n\end{array}\right.
\]
The above definition is not that intuitive, but basically, the counting process $N_t$ tells us how many events have occurred up to time $t$. To see this, we can also define the inter-arrival times between events $\tau_n:=T_{n}-T_{n-1}$ for $n\geq1$ and define $\tau_0:=0$. Thus, $N_t$ can be written as

\[
N_t=\max\limits_{n\geq0}\{\tau_0+\tau_1+\ldots+\tau_n\leq t\}.
\]
Many of the point processes that are found in the literature are simple, which means that only one event arrives each time. That is,
\begin{equation}\label{def:simple}
\lim\limits_{h\to0}\Px\left[\frac{N(t+h)-N(t)>1}{h}\right]=0
\end{equation}
\noindent A point process is called \emph{regular} if it possess an intensity function $\lambda(t)$, which is defined, when the limit exist, as
\[
\lambda(t)=\lim\limits_{h\to0}\Ex\left.\left[\frac{N(t+h)-N(t)}{h}\right|\sigma(N_s)_{0\leq s\leq t}\right],
\]
\noindent where $\sigma(N_s)_{0\leq s\leq t}$ denotes the past history of the process $(N_t)_{t\geq0}$. Further, whenever the point process is simple, the above definition is equivalent to,
\[
\lambda(t)=\lim\limits_{h\to0}\Px\left.\left[\frac{N(t+h)-N(t)>0}{h}\right|\sigma(N_s)_{0\leq s\leq t}\right]
\]
Notice that by definition, as $\lambda(t)$ is a conditional expectation over a sigma algebra, it is itself a random process, but sometimes, it can be a deterministic function of time as in the case of the Poisson process. Since probabilities are always non-negative, the intensity function is always non-negative and thus the cumulative intensity function, $\Lambda(t)=\int_0^t \lambda(s) ds$ is a non-decreasing function. The cumulative intensity function always exists (even if the point process is not regular). In fact, it can be shown that the cumulative intensity function is the compensator of the point process $N(t)$. 

A technical remark might be in place here: since a temporal point process is always a càdlàg 
(i.e., right continuous with left limits) submartingale, it can \textit{always} be then decomposed as $N(t)=M(t)+A(t)$, where $M$ is a local martingale and $A$ is an increasing predictable process called the compensator. As can be inferred by the name, it turns out that $\Lambda(t)=A(t)$. Moreover, whenever the measure generated by $\Lambda(t)$ is absolutely continuous with respect to the Lebesgue measure, the ``normal'' intensity function $\lambda(t)$ exists and so the limit above exists.

\sk

To explain the above definitions and before proceeding to the canonical example of interest we can look at the simplest and most common point process: the Poisson process. In this case, the intensity is completely deterministic, i.e. $\lambda(t):[0,\infty)\to[0,\infty)$ and
\vspace{-.3cm}
\[
\Px[N(t)-N(s)=k]=e^{-\int_s^t \lambda(u)du}\frac{1}{k!}\left[\int_s^t \lambda(u)du\right]^k.
\]
Since the intensity is deterministic we always know the \textit{infinitesimal} probability of a new event. That is, the probability that we observe a new event in the interval $[t,t+h)$ is roughly $\approx h\lambda(t)$ and in this case this probability is actually independent of the process' history. In this case, it can be shown that $ \Ex[N_t]=\int_0^t \lambda(s)ds$.

\sk

The Poisson process and many others point processes are good to model events where the rate at which new events appear is independent of how many or when past events occurred, but sometimes the phenomena that wants to be analyzed exhibits a ``clustering effect'', where the appearance of one new event will trigger the occurrence of more new events. An important example that possess this property, and a process that will be used thoroughly during this proposal, is the family of Hawkes processes, which was first studied thoroughly by Ogata (see \cite{ogata1983estimation}) to model the occurrence of earthquakes in Japan. His motivation to use this type of processes was that usually after a big earthquake, many replicas will follow. In the Hawkes Processes family, the (stochastic) intensity function, $\Lambda(t)$, ``feeds'' itself from the point process $N_t$. That is, the intensity increases when a point arrives, which in turn will trigger more points to arrive. In the most classical example, the intensity can be characterized by\\[-.4cm]
\begin{equation}\label{eqn:hawkes:int}
\Lambda(t)=\lambda_0+\int_0^t\phi(t-s)dN_s=\lambda_0+ \sum\limits_{T_i<t}\phi(t-T_i),
\end{equation}
\noindent where $N_t=\{T_1,T_2,T_3,\ldots,T_{N_t}\}$ is itself the random point process and $T_i$ is the time of the $i-$th occurrence. As mentioned before, the intensity is also stochastic (random) but the source of its randomness is not exogenous but comes from the point process itself. Moreover, if $\phi(t)=Q e^{-r t}$, then the process $(N_t,\lambda_t)$ is a Markov Process (a process whose evolution depends solely on the current state the system and not in the history of \textit{how} it arrived to such current state) and the intensity function decays exponentially between events. As before, we can also compute the expected number of events in the interval $[0,T]$ which turns out to be $\Ex[N_t]=\frac{\lambda_0}{1-\int_0^{\infty}\phi(s)ds}$.

\begin{figure}[ht]
\centering
		\includegraphics[width=0.69\textwidth]{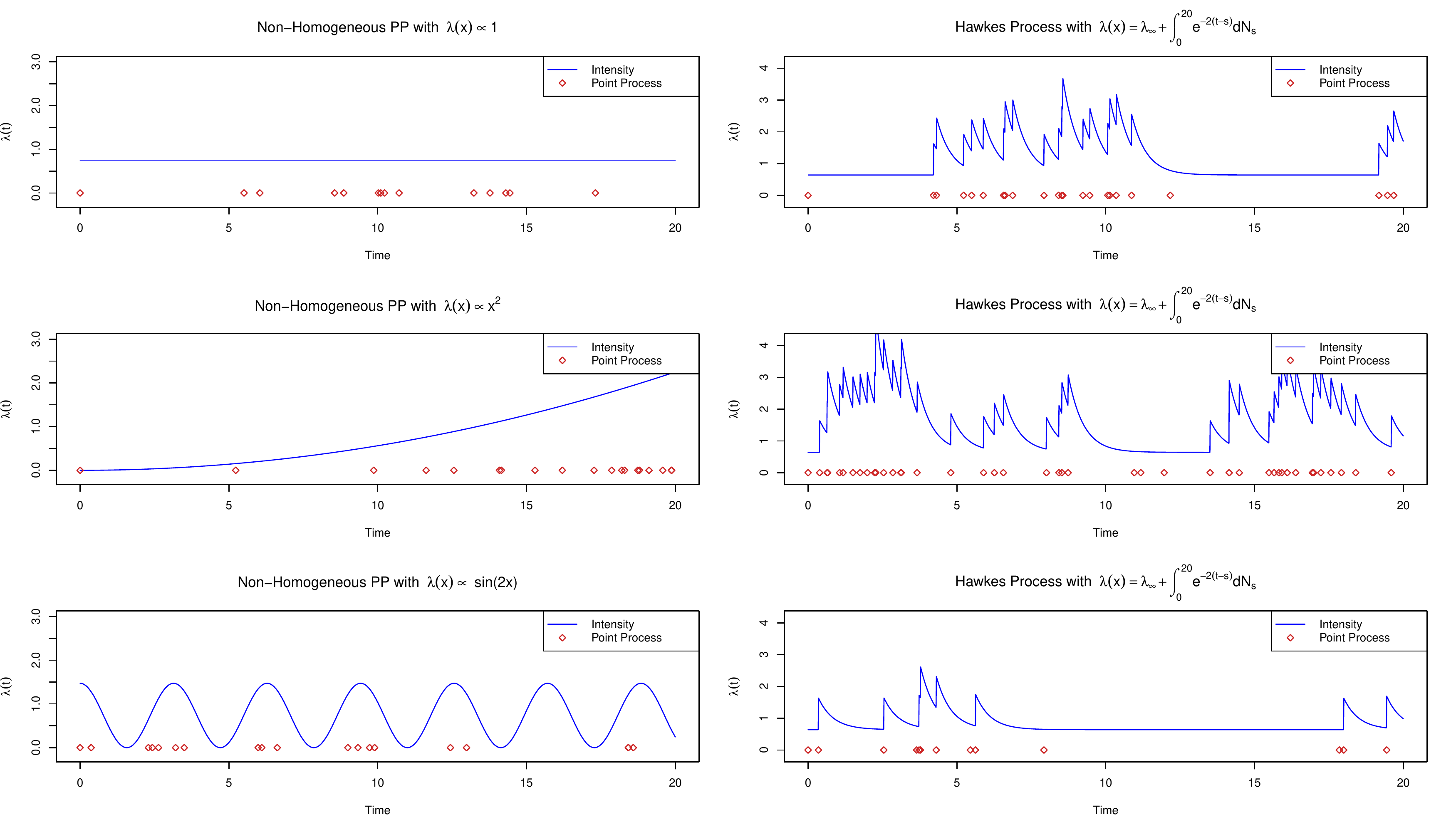}
	\caption{Six realizations of point process with $T=20$ and $\Ex[N_T]=25$. On the left, three different inhomogeneous Poisson processes with intensities proportional to a constant, $x^2$ and $\sin(2x)$ are shown, whereas on the right three realizations of \textit{the same} Hawkes process are displayed.}
	\label{fig:examples}
\end{figure}

\sk

Figure \ref{fig:examples}, tries to illustrate the concepts and processes described above. In all the plots the intensity function and a realization of the point process are shown. At the same time, all the processes below were ``standardized'' by setting the \textit{expected} number of points to be the same in all cases. The left column shows 3 different inhomogeneous Poisson processes with different intensities. Since the intensity is deterministic, it does not matter how many simulations of the point process are performed, the intensity will the the same. As it can be seen, when the intensity is higher, the likelihood of more points increases. In the right column, 3 different realizations of the same Hawkes process are plotted. In this case since the intensity is stochastic, it will change between realizations, but more importantly the \textit{clustering} phenomena can be observed. Whenever a new event arrives, the intensity spikes and thus the likelihood of more events happening increases and to balance the process, in between the arrival of events the intensity decreases exponentially reducing the likelihood of more events coming. However, there is always a baseline intensity for which the intensity process cannot go below guaranteeing that there will be a new event at some point. To illustrate this, we present next a heatmap of the Hawkes process (i.e. the process counting the number of events up to time $t$) and one of its intensity.

\begin{figure}[h]
\centering
\begin{subfigure}{.5\textwidth}
  \centering
  \includegraphics[width=.9\linewidth]{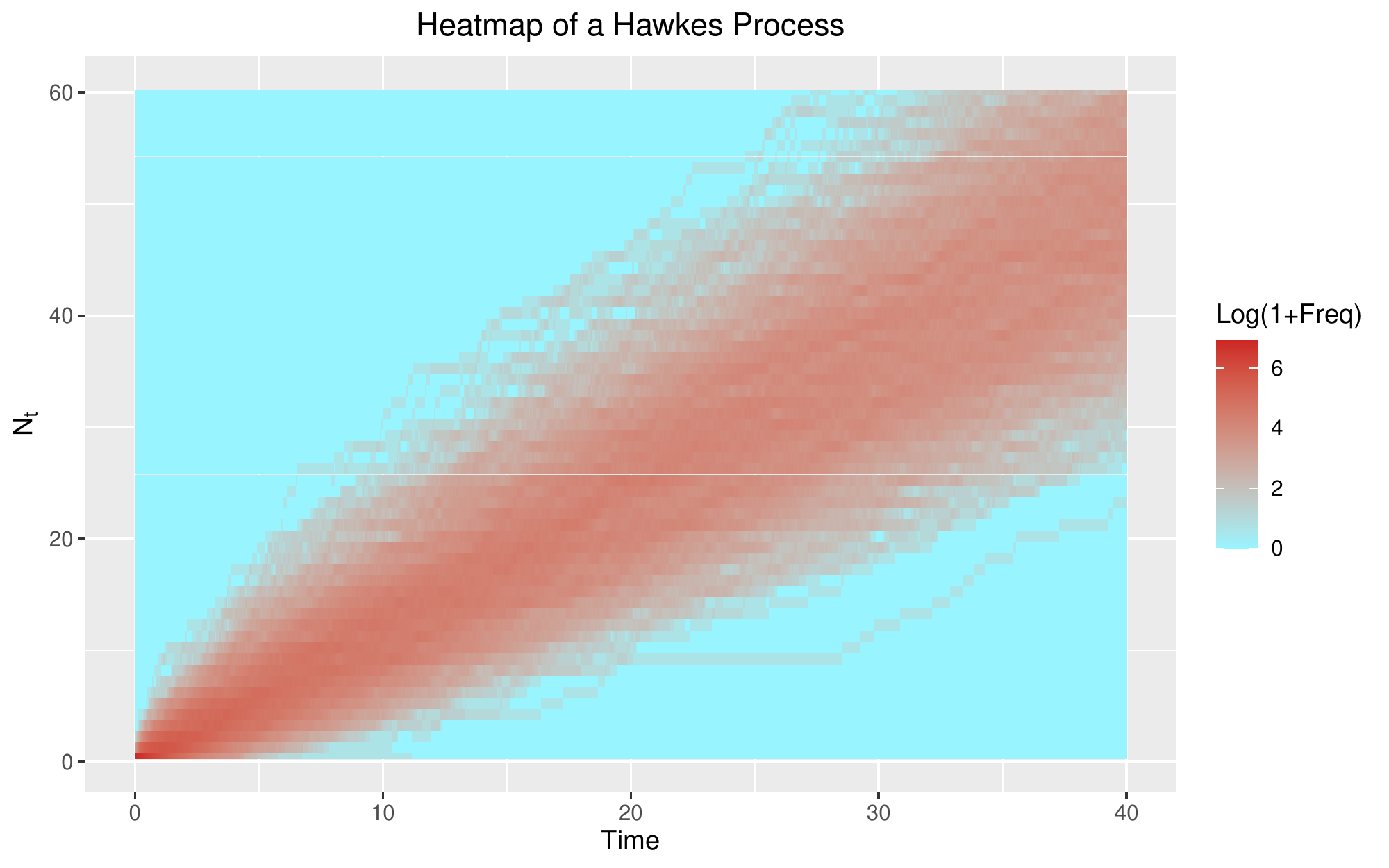}
  \caption{The Hawkes process}
  \label{fig:heatmap:N}
\end{subfigure}%
\begin{subfigure}{.5\textwidth}
  \centering
  \includegraphics[width=.9\linewidth]{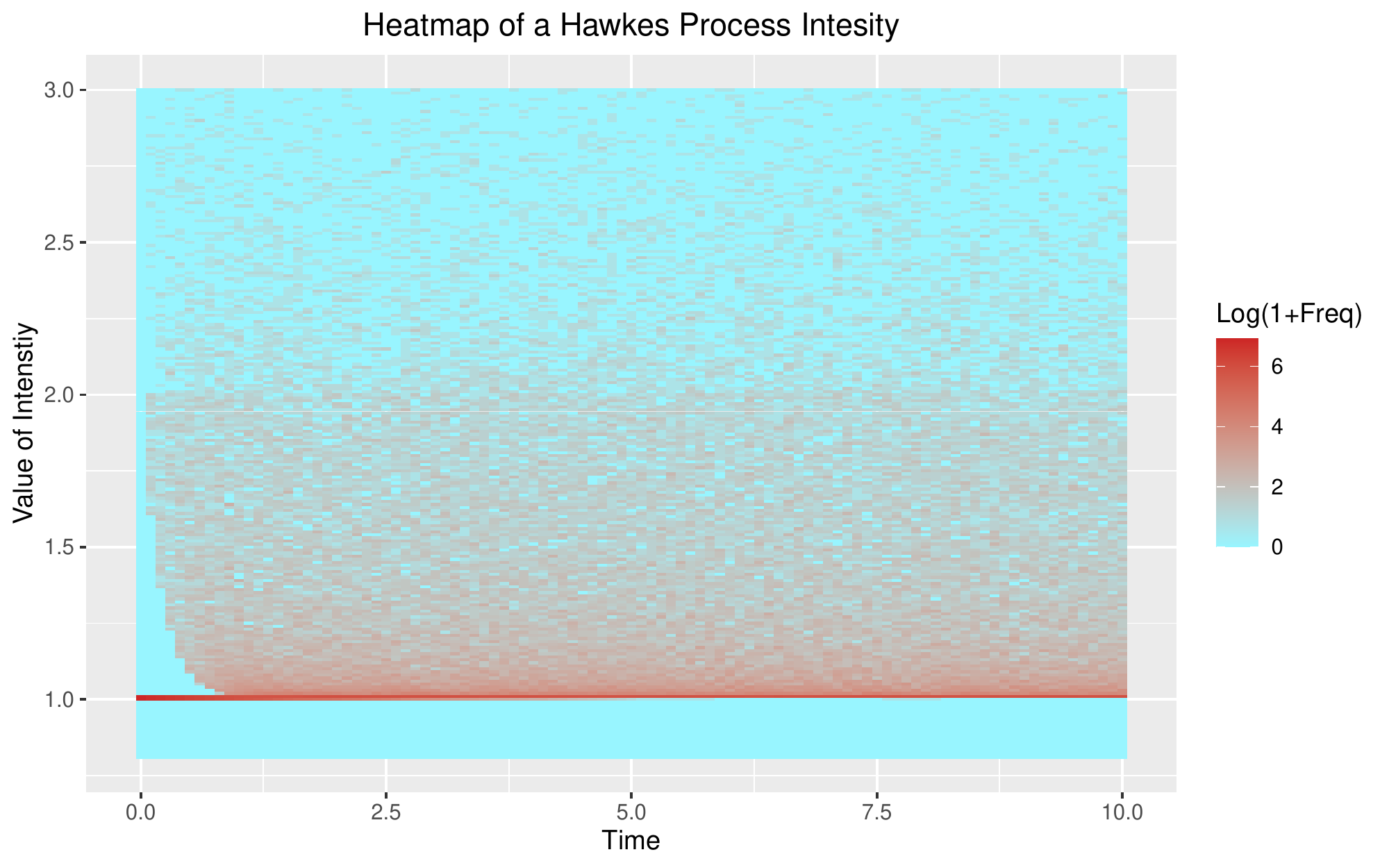}
  \caption{The intensity process}
  \label{fig:heatmap:I}
\end{subfigure}
\caption{Heatmap of the Hawkes process and its intensity. The frequency is given in a logarithmic scale.}
\label{fig:heatmaps}
\end{figure}

\subsection{Linking point processes to Epidemiology}

Epidemic models have been trying to model, understand and predict different features and generalities of the pandemics that humanity has lived with. However, it was not until around 1930 that the first stochastic epidemic model was created. Before that, all the models were deterministic were based on creating different systems of ODEs, but some desirable random effects were missing. Nonetheless, those deterministic models provided the skeleton upon which the behavior of corresponding stochastic systems are built. Furthermore, many times, when populations are large, the limiting behavior of such stochastic systems converge to a corresponding system of ODEs. The main models discussed in this proposal will be variations of the famous stochastic susceptible-infected-recovered (SIR) models and belong to the class of the so-called compartmental models, where each individual is placed into one of the compartments and the individuals move between the compartments under predetermined rules. For a survey on the description of many stochastic models, see \cite{greenwood2009stochastic}.

\sk

The most common stochastic epidemic model, and the building block for many of the more detailed and precise models produced nowadays, is the so-called \textit{general stochastic model},

\begin{itemize}
    \item There are three compartments where an individual can be: Susceptible, Infected and Recovered. Moreover, each individual belongs to one and only one compartment.
    \item The number of susceptible, infected and recovered individuals at time $t$ are denoted by $S_t,I_t,R_t$. Also, the population is constant at a level $N$ and thus, $N=S_t+I_t+R_t$.
    \item A recovered individual cannot contract the disease again.
    \item During the interval $[t,t+\Delta t]$, one of the following things has to happen:
    \begin{enumerate}
        \item A susceptible individual gets infected with probability $\beta\frac{I_t}{N}$ because $I_t/N$ is the likelihood of contacting one of the current number of infected individuals, $I_t$, multiplied by the likelihood $\beta$ of contracting the virus once you got in contact with a person. Since there are $S_t$ susceptible people at that moment,
        \begin{equation}\label{eqn:St} 
        \Px\Big[(S_{t+\Delta t},I_{t+\Delta t})-(S_t,I_t)=(-1,1)\Big]=\beta S_t\frac{I_t}{N}\Delta_t + o(\Delta t)
        \end{equation}
        \item An infected individual recovers from the virus at any time with likelihood $\gamma$. Thus, since there are $I_t$ infected individuals with the same likelihood to recover,
        \begin{equation}\label{eqn:Rt}
        \Px\Big[(S_{t+\Delta t},I_{t+\Delta t})-(S_t,I_t)=(0,-1)\Big]=\gamma I_t\Delta_t + o(\Delta t)
        \end{equation}
        \item Nothing happens and the system remains the same. This is the complementary event to the union of the two actions above,
        \begin{equation}\label{eqn:nothing:SIR}
        \Px\Big[(S_{t+\Delta t},I_{t+\Delta t})-(S_t,I_t)=(0,0)\Big]=1-\left(\beta \frac{S_t}{N} +\gamma\right)I_t\Delta_t +o(\Delta t)
        \end{equation}      
    \end{enumerate}
\end{itemize}

Many of the compartmental models that are available nowadays, to allow for an easy simulation and mathematical tractability, remain within the framework of the Markovian world. As such, those models have deterministic limits and diffusion approximations. The main results and techniques to obtain such limits are detailed in \cite{kurtz1981approximation,meyn}. These limits can be used to estimate and calibrate the models quickly and constantly which can be of great help for public policy makers, since many public health policies might be influenced by predictions of how large an epidemic might be. 

\sk

In a more abstract form, the model (\ref{eqn:St}-\ref{eqn:nothing:SIR}) above can be written in terms of the difference of some Poisson processes and in fact, at their core, many stochastic SIR-type models can be established as functions of Poisson processes or as limits of them, which leads to an immediate question of \textit{whether the point process leading the epidemic model can be generalized to include more features observed empirically in the data}. Besides, there have been few models that try to explore the dynamics of the frequency at which individuals get infected and its effects in the likelihood of future people getting infected. That is, many of the current models either try to elaborate on the number of compartments or in the finesse of the conditions required for a certain individual to transition from one compartment to another but there are few models that try to get a more realistic description of when such transitions occur and how the frequency of such transitions affects the future evolution of the system. 

\sk

Indeed, in most of the current compartmental models such as (\ref{eqn:St}-\ref{eqn:nothing:SIR}), each susceptible individual is equally likely to get infected and that probability is solely dependent on two factors: the proportion of infected individuals and the number of susceptible individuals at time $t$. To understand better paradigm assume the following two scenarios

\begin{itemize}
    \item Assume a population of a 100 individuals and at time 0, there are 5 infected people and 95 susceptible but at time $t=20$, it is observed that $(S_{20},I_{20},R_{20})=(50,30,20)$.
    \item Assume the same population as above with the same initial conditions of $(S_0,I_0,R_0)=(95,5,0)$ but at time $t=5$, it is observed that $(S_{5},I_{5},R_{5})=(50,30,20)$.
\end{itemize}

As it can be seen, in both cases you start with the same number of infected and susceptible people and at at some point in time you also have the state $(50,30,20)$. The difference is that in the first case it took longer to reach that point, whereas in the second, such state was reached much faster. When this happens it might be better to assume that the driving process counting the number of infected people can cluster due to the fact that in periods where there are many infected people, the likelihood of more people to get infected is higher than in periods where few people have been recently infected. In this spirit, various empirical studies have revealed that when an arrival processes shows display \textit{over-dispersion}, i.e. the variance of the number of arrivals in a given interval exceeds the corresponding expected value, the standard assumption of having a Poisson process driving the number of infections is not valid and new models are required. As Figure \ref{fig:overdispersion} suggests, COVID-19 has proven to be a very contagious virus where a high level of over-dispersion could be observed. This implies that rather than having a stochastic epidemic model featuring Poisson processes, a Hawkes point process might be a much more suitable candidate to model it.

\begin{figure}[H]
\centering
		\includegraphics[width=0.83\textwidth]{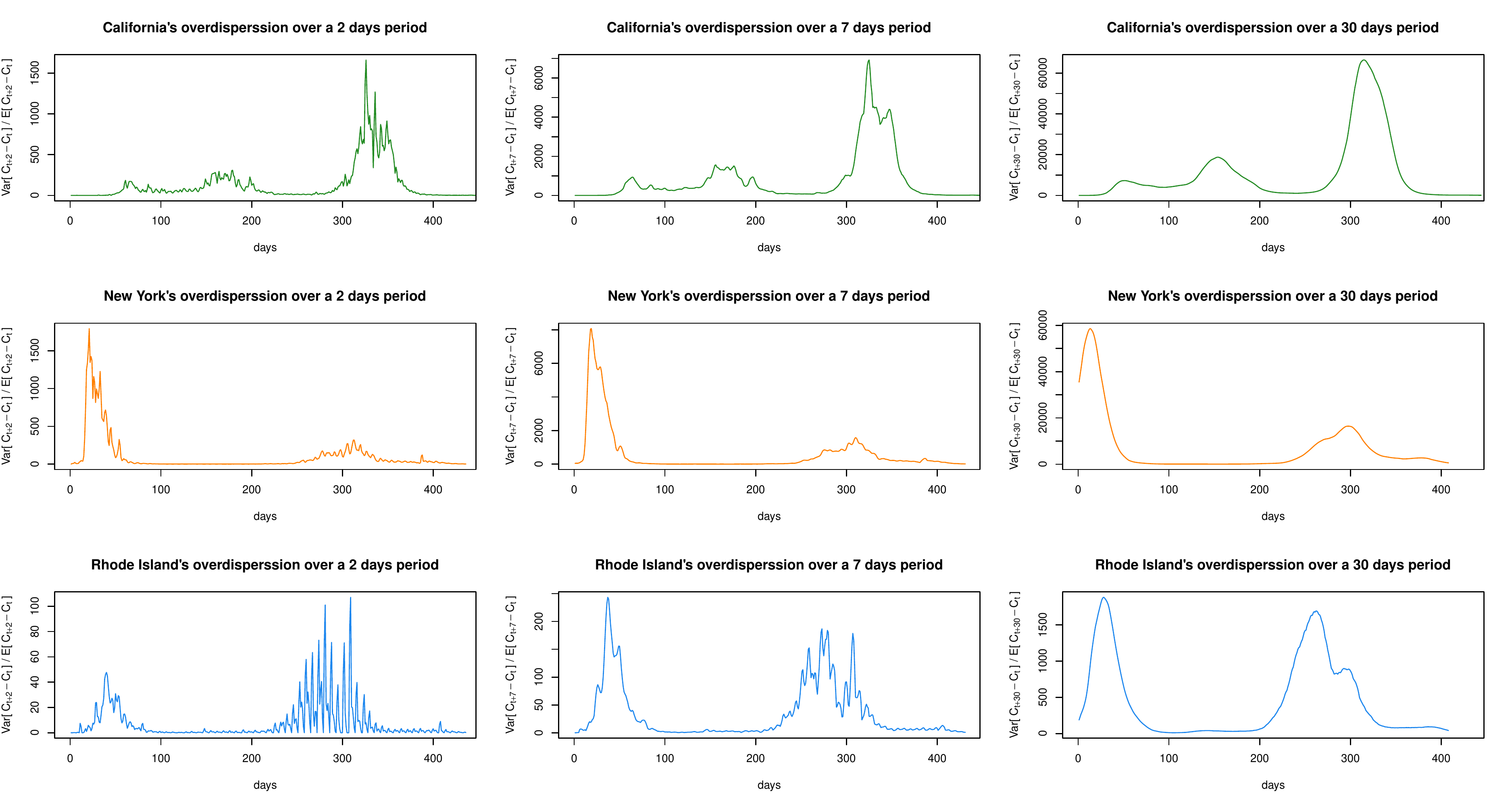}
	\caption{Plots showing the over-dispersion in a 2,7 and 30 days window of the number of newly infected people during COVID-19 in 3 US states: New York, California and Rhode Island. The $x$ axis shows the time in days, starting from February of 2020, while the $y$ axis shows the quotient between the variance and the average of the new cases within that time period. A quotient greatly larger than 1, indicates high over-dispersion in the ``arrival process'' of new infect people.}
	\label{fig:overdispersion}
\end{figure}

Our first step will be the creation of a model that generalizes the Poissonian flow. In contrast to deterministic models, rather than modeling the number of infected people, $I_t$, directly the contagion process $C_t$ (counting the number of infected people up to time $t$) is the process that needs to be modeled. Reason being that $I_t$ can increase (with a new infection) and decrease (with a recovery) but the counting process associated to a point process must be non-decreasing.

\sk

One of the first attempts of incorporating Hawkes processes into an SIR model is provided in \cite{rizoiu2018sir}. In there, the authors explain a relationship between the general stochastic SIR model and a finite population Hawkes-SIR model (they account for the fact that the population is capped at $N$ and therefore $C_t\leq N$), but this model has some limitations. In particular they only explore a model where the \textit{conditional expectation} of the intensity of $C_t$ conditioned on \textit{all} the times at which there is a recovery coincides with a finite Hawkes process with a baseline intensity $\lambda_0(t)=0$. These restrictions cripples the model in the sense that a true Hawkes-SIR model would require the intensity of $C_t$ to be of the form \eqref{eqn:hawkes:int} with $\lambda_0(t)\neq0$ and also not be the result of a conditional expectation on past and future unknown times.

\sk

Other models that have considered an epidemic and relating the total number of infected people with a Hawkes process are presented in \cite{escobar2020hawkes,garetto2021time,chiang2020hawkes} among others. Nonetheless, there is missing a model that becomes a natural extension of an stochastic SIR model. That is, there is no model in the literature where people are divided into the $S,I,R$ compartments and people start arriving to the ``infected'' compartment according to a Hawkes process. Recovered individuals should arrive via a Poisson or renewal process, since contrary to infections, the recovery of a person does not affect how other people recover. One of the aims of this proposal is to create and analyze such model as well as find conditions under which a disease introduced into a community will develop into a large outbreak, and if it does, conditions under which the disease may become endemic.  This condition is linked to the so-called basic reproductive number, $R_0$, defined as the expected number of secondary infective cases per primary cases in a susceptible population. 

\sk

As it can be noted, many of the models that have started to incorporate Hawkes processes are very recent and inspired by the highly infectious rate of COVID-19, making this a novel area to propose, analyze and create models that can shed some light into the question of how fast this disease propagates, hopefully contributing to the decision process on policy-making by providing more accurate estimates when incorporating some of the salient feature of these types of models.

\textbf{Scope and Limitations of the current model:} In its current form, the model presented in this article provides a way to model the overdispersion and high variability observed in the data. However, due to the limiting ability of incorporating a stronger dependency on the current state of the epidemic and the different rates at which different variants have shown to be transmitted, this model is more suitable for keeping track of the pandemic more at the beginning stage when not so many social and political reactions have occurred. To account for this, the author is currently investigating a regime switching alternative. Further, as it is, these type of models should only be used to forecast a short to medium term prediction.

\textbf{Organization of this article:} In section \ref{Sec:Model}, a brief description of the model is presented while in Section \ref{Sec:Quantities} we compute the moments of the process which can be used to predict the medium term dynamics of the current number of infected people given the probabilistic distribution of the quarantined people (which in terms can be used for public policy). Section \ref{section:estimation} provides a brief survey and ideas of how to perform efficient estimation on a state-dependent self-exciting process. However, since estimation is out of the scope of this paper, it will be treated in subsequent research. In Section \ref{Sec:Numerical}, we provide different numerical examples to explore how the different parameters of the model play a role in it and their sensitivity. The conclusions and elements of further research are presented in Section \ref{Sec:Conclusions}. Finally two appendices are also provided. In Appendix \ref{Sec:Proofs} all the proofs to the technical lemmas and theorems of Section \ref{Sec:Quantities} are provided while in Appendix \ref{Sec:algorithms} the pseudo-code of the different algorithms used to simulate the different processes are provided.

\section{A Brief Description of The Model} \label{Sec:Model}

In the model presented in this article, an important assumption will be that the  the population is rather large (technically, it is assumed that the possible number of susceptible people is infinite to facilitate the analysis of the moments of the number of infected people $I_t$). This assumption provides a tractable background for the problem and is not extremely unrealistic given that the model is specifically designed to simulate a pandemic in its early stages. Considering that up May 1st, 2022 only 6.6\% of the population has been reportedly infected\footnote{Information taken from \url{https://covid19.who.int/}} (although it is of general consensus that this number is very underestimated, specially now with the development of at-home test kits), we can think of the number of susceptible individuals as a rather large number and the model will still be accurate. In fact, a precise finite model can be created, but there is little to gain and the closed formulas obtained in Section \ref{Sec:Quantities} become very cumbersome. This compartmental model follows a similar SIR dynamics than the models in \cite{greenwood2009stochastic}, where $S_t, R_t$ and $I_t$ represent the number of susceptible, recovered and infected people at time $t$. Obviously, since, as discussed above, the number of susceptible people is assumed to be infinite, we cannot have a classic relationship such as $S_t+I_t+R_t=N$, but we will rather only focus on the number of infected individuals $I_t$ and the number of recovered individuals $R_t$. 

\sk

The main modeling assumption in this paper will be that the (historical) number of infected people up to time $t$, denoted by $C_t$, will be driven by a (Hawkes-like) counting process with stochastic intensity 
\begin{equation}\label{intensityC}
\Lambda_t=\lambda^\infty(I_t,R_t)+\int_0^t Q_s\phi(s,t)dC_s,
\end{equation}
where $\{Q_n\}_{n=1}^\infty$ is a sequence of i.i.d. random variables independent among themselves and from every other process. As described below, these random variables $Q$ will be interpreted as the level of quarantine that the $n-$th infected person will have, modulating the probability that this individual will produce future infections.

\sk

In this paper we will assume that the baseline intensity $\lambda^\infty(i,\rho)$ of the Hawkes-like process $C_t$ given in Equation \eqref{intensityC} reacts to the state of the epidemics at time $t$. Indeed, we will assume that if $I_t=i$ and $R_t=\rho$. Then, 

\begin{assumptions}\label{assump:baseline}
    \lambda^\infty(i,\rho)=\lambda_0+i\log(\alpha)+\rho\log(\beta),
\end{assumptions}
where $\alpha>1$ and $\beta,\lambda_0$ are positive numbers. 

\sk

The form of the baseline intensity given in assumption \ref{assump:baseline} provides great flexibility to consider multiple situations. We describe some of them next.
	\begin{itemize}
		\item[$\bullet$]  \textsl{A pure Hawkes approach.} In this setting, the intensity takes the form of a classical Hawkes process. This could be a basic model to use for explaining the over-dispersion observed in the data. This can be achieved by setting $\alpha=\beta=1$. In this sense, the baseline intensity will be independent of the number of infected and recovered and the general intensity of the point process will only depend on the cumulative number of infected people and times of infection of these. In this case, the intensity becomes
\begin{equation}\label{eqn:purehawkes}
\Lambda_t=\lambda_0+\int_0^t Q_s\phi(s,t)dC_s.
\end{equation}
This case is considered in \cite{lesage2020hawkes} where the author analyzes, simulates and fits a classical Hawkes process for the total number of infected people in France but to signify a change in public policy modifies the intensity to mimic a lock-down as in China. As another example, in \cite{escobar2020hawkes} where the author discretizes the intensity to analyze and fit the epidemic model to the gathered data of several countries, including Mexico. This model is very popular because there are well established algorithms to simulate and to calibrate the model to existing data.
\end{itemize}

	\begin{itemize}
	\item \textsl{A proportional-to-infections approach}. In this case, the baseline intensity will be  proportional to the current number of infected people people. Here, not only can over-dispersion be captured but the intensity, and thus the likelihood, of observing a new infection
will increase or decrease according to current number of infected people. This can be achieved by setting $\lambda_0 = 0$ and $\beta = 1$ and the intensity becomes
\begin{equation}\label{eqn:proptoI}
\Lambda_t=\lambda_0+i\log(\alpha) + \int_0^t Q_s\phi(s,t)dC_s,
\end{equation}
	\item \textsl{A state dependent Hawkes baseline intensity}. This is the more general case and the one that will be considered in this paper. In this case, assumption \ref{assump:baseline} can be used to focus on different aspects of the model depending on whether $0<\beta<1$, $\beta=1$ or if $\beta>1$. Indeed, if $0<\beta<1$, then the quantity $\rho\log(\beta)$ is negative and in this case we can think of a situation where the recovered individuals suppress the epidemics. This is more in line with the classical assumptions of an SIR model where by increasing the number of recovered people the infectivity of the virus decays and new cases are less likely. However, this behaviour hinges on a critical assumption: recovered people cannot contribute more to the epidemics and that they become isolated from it. This assumption is not completely true for the COVID-19 epidemics where many recovered people become reinfected and can contribute to the spread of the disease. One should be careful here though, if the number of recovered individuals become very large and $\beta<1$, the baseline intensity $\lambda^\infty(i,\rho)$ might become negative and this is not permitted. Several solutions could be considered at this point, the two more relevant being to stop the process before $\lambda^\infty(i,\rho)<0$ or change the form Assumption \ref{assump:baseline} to $\lambda^\infty(i,\rho)=\bigg(\lambda_0+i\log(\alpha)+\rho\log(\beta)\bigg)_+$, where $f(x)=(x)_+$ represents the positive part of $x$, being this latter option a very interesting case left out for further research. If $\beta=1$, then the number of recovered people does not affect the model and the process $R_t$ becomes a homogenous Poisson process independent of $C_t$. Finally, whenever $\beta>1$ the recoveries can also contribute to the infection of new individuals (rather than suppressing it). As discussed above, in the current pandemic many individuals can become reinfected and keep transmitting the virus to to other people. This feature might comes in handy specifically at the current situation with COVID-19 as universities and other workplaces treat recovered individuals as non-infectious but indeed they might become infectious again but at a different rate than that of susceptible or currently infected individuals.
\end{itemize}

It is also important to determine the parameters of the model:

\begin{itemize}
	\item The baseline infective rate $\lambda_0$, which is associated to the likelihood that someone new gets infected due to exogenous factors to the model such as migrations, population dynamics, etc. This is not affected by the amount of infected or recovered people.
	\item The infective rate $\alpha>1$. This parameter (or more precisely, the logarithm of it) measures how the likelihood at time $t$ of a new infection will increase due to the number of infected people at that particular instant regardless of the total number of people that has been infected by the virus. The higher $\alpha$, the more likely a new infection will occur at the level of current infections $I_t$.
	\item The situational rate $\beta>1$. This parameter (or more precisely, the logarithm of it) measures how the likelihood at time $t$ of a new infection will increase or decrease due to the number of recovered people \textsl{up to} that particular instant. As $\beta\to0$ the likelihood that a new infection will occur at the level of current of $R_t$ decreases while if $\beta\to\infty$ the likelihood will increase.
	\item The mean recovery time $\mu$. This parameter is the expected recovery time period from the infection. There are several studies about this quantity and it depends on several factors, but according to \cite{chowdhury2021comparative} a good approximation could be around 9 days. 
	\item The historical influence parameter $r$ (also called the reversion coefficient $r$). This parameter will measure the influence of previous (historical) infections in the arrival of a new infection and indicates the speed at which the likelihood of a new infection decays to the current level of the baseline intensity on the absence of a new infection. In fact, for out model, the increase in the likelihood of a new infection due to a previous infection at time $s<t$ is proportional to $e^{-r(t-s)}$. This means that more recent an infection, the more it will contribute to the likelihood of a further infection. This behavior is the key difference with respect to the classical epidemiological models and what allows for the clustering and \textsl{over-disperssion} observed in the data for the current pandemic.	
	\item The probability distribution of the ``quarantine effect'' $Q$. This random variable specifies the level of ``quarantine'' each individual will have. Indeed, our model specifies some stochastic dynamics where the probability of a new infection at time $t$ is driven by the amount of people that has been infected up to that time and how recent their infections have been. This quarantine factor will basically determine the proportion at which infected individual person will contribute to a new infection. From the modeling perspective, this random variable can be thought as a measure of quarantining. The lower the (random variable) $Q$, the lower the contribution of such infected individual to a new infection. This random variable can be discrete or continuous, but its support has to be over the positive numbers. In other words, an infected individual must have a positive contribution to the general likelihood of generating a new infection even if it is small That is, there cannot be a ``perfect quarantine'', which is consistent because people have to go to the groceries or buy basic services and even interact with delivery services by receiving goods at home. For this work we will impose the mild restriction that the Moment Generating Function (MGF) of $Q$ exists on a neighbourhood of 0.
	
	\sk
	
	\begin{rem}
	It is also known (see Chapter 3.3 in \cite{laub2022elements}) than when the random variable $Q$ is constant, and the intensity of the point process is the classical Hawkes process in Equation \ref{eqn:hawkes:int} with $\phi(t)=Q e^{-r t}$, then the \textsl{branching ratio} of the process would be given by 
	\[
	n=\frac{Q}{r}
	\]
	where as above, $r$ is the reversion coefficient. In the SIR process where the number total (historical) number of people that has been infected by the virus is driven by a Hawkes process the branching ration $n$ has the interpretation that when $0<n<1$, it becomes the ratio of the number of people that one individual will infect relative to the entire population; that is, it is related to the epidemiological basic reproduction number $\mathcal{R}_0$.
	\end{rem}
	\sk
\end{itemize}

%

\begin{rem}
Assumption (\ref{assump:baseline}) specifies the form of the so-called baseline intensity. This quantity will not depend on the past number of infections but solely on the present number of infected and recovered individuals. In fact, this baseline intensity remains constant between events, that is, between new infections or recoveries, the probability of a new infection is exponentially distributed with rate $\lambda^\infty(i,\rho)$ and thus the process dictating the arrival of a new infection is equal in distribution to a homogeneous Poisson process with rate $\lambda^\infty(i,\rho)$. 

\sk

 Also, notice that $\alpha$ and $\beta$ represent the factors that drives the baseline intensity according to the state of the system. The more infected people there is at the moment, the more likely a new infection will occur, and the more recovered individuals there are in the present moment, the less likely a new infection will occur.
\end{rem}

Further, since we don't want a ``degenerate'' Hawkes process, we will assume that $Q$ does not have an atom at 0. That is,
\begin{assumptions}\label{assump:Q}
\Px[Q\leq 0]=0
\end{assumptions}

\section{Derivation of the Moments of the Number of Infected individuals $I_t$} \label{Sec:Quantities}

As mentioned in the previous section, one of the features that the epidemic model presented her has is that every person is assumed to have a random level of quarantine, whose law is given by $Q$. As an easy example, assume that $supp(Q)=\{0.25,0.75,1\}$. In this case, $Q=0.25$ will imply a higher level of quarantine (contributing less to new infections) while $Q=1$ would mean a low level of quarantine (so that this person will increase the likelihood of a new infection happening). Obviously, $Q=0.5$ would be an intermediate case. 

To compare how this different levels of quarantine and other parameters affect the model, it might be worth to look at the average behaviour of $I_t$ and $R_t$. This is because depending on Var$[I_t]$, some comparisons might not be depicted accurately by plotting some trajectories of the process.

Recall from Equation \eqref{intensityC} that the total number of infected people up to time $t$, denoted by $C_t$, is a determined by a counting process with stochastic intensity given by 
\begin{equation}
\Lambda_t=\lambda^\infty(I_t,R_t)+\int_0^t Q_s\phi(s,t)dC_s,
\end{equation}
where $\{Q_n\}_{n=1}^\infty$ is a sequence of i.i.d. random variables independent among themselves and from every other process denoting the level of quarantine that the $n-$th infected person will have, modulating the probability that this individual will produce future infections.

We are interested in computing the generating function of the triplet $(I_t,R_t,\lambda(t))$
\[
\Ex\left[z^{I_t}w^{R_t}e^{-s\lambda_t}\right]
\]
Here, we will assume that all the stochastic processes are Markovian, which can be attained if
\begin{itemize}
    \item The recovery time is exponentially distributed with rate parameter $\mu$. That is, if $\tau$ is the recovery time for an infected individual, then 
\begin{assumptions}\label{assump:tau}
    \Px[\tau\leq t]=1-e^{-\mu t}.
\end{assumptions}
This implies in particular that if $I_t=k$, then 
\begin{equation}\label{intensity:R}
\Px[R_{t+\Dt}-R_t=1]=\Px[\min\{\tau_1,\ldots,\tau_k\}\leq \Dt]=1-e^{k\mu\Dt}=k\mu\Dt+o(\Dt)
\end{equation}
    \item The self-exciting kernel $\phi(\sbt)$ is an exponential function. That is,
  \begin{assumptions}\label{assump:phi}
    \phi(s,t)=e^{-r(t-s)}
\end{assumptions}  
\end{itemize}

These assumptions can be relaxed, but the analytical tractability will be lost and different techniques would have to be employed. The analysis presented here is inspired by the one presented in \cite{koops2018} but here is generalized to allow a state dependent baseline intensity.

In order to compute the joint distribution of $(I_t,R_t,\Lambda(t))$, we need to solve a system of differential equations presented next.

\begin{thm} \label{thm:pde:f}
Let $F(t,i,\rho,\lambda)=\Px[I_t=i,R_t=\rho,\Lambda_t\leq \lambda]$ and $f(t,i,\rho,\lambda)=\dfrac{\partial F}{\partial \lambda}(t,i,\rho,\lambda)$. Then, under assumptions (\ref{assump:tau})-(\ref{assump:phi}), $f(t,i,\rho,\lambda)$ satisfy the Partial difference differential equation (PDDE):
\begin{equation}\label{pde:f}
\begin{aligned}
\dfrac{\partial }{\partial t}f(t,i,\rho,\lambda)-\frac{\partial}{\partial \lambda}(r\lambda& f(t,i,\rho,\lambda))+r\lambda^\infty(i,\rho)\frac{\partial}{\partial \lambda} f(t,i,\rho,\lambda)=\\& 
\int_0^{\lambda-\Delta\lambda^{\infty,1}_{i,\rho}}yf(t,i-1,\rho,y)\frac{\partial}{\partial \lambda}\Px[Q\leq \lambda-\Delta\lambda^{\infty,1}_{i,\rho}-y]dy+\\
 & (i+1)\mu f(t,i+1,\rho-1,\lambda-\Delta\lambda^{\infty,2}_{i,\rho})-(i\mu+\lambda)f(t,i,\rho,\lambda)
\end{aligned}
\end{equation}
where
	\[
	\Delta\lambda^{\infty,1}_{i,\rho}=\lambda^\infty(i,\rho)-\lambda^\infty(i-1,\rho) \qquad\qquad\text{and}\qquad\qquad	\Delta\lambda^{\infty,2}_{i,\rho}=\lambda^\infty(i,\rho)-\lambda^\infty(i+1,\rho-1)
	\]
\end{thm}

The next objective is to compute the generating function of the triple $(I_t,R_t,\Lambda(t))$, which can then be used to compute the moments and other quantities of interest. However, since $I_t$ and $R_t$ are discrete and $\Lambda(t)$ is continuous, we will have to compute the transformations of those random variables separately.

\begin{thm} \label{thm:pde1} Let $\varphi(t,z,w,s)=\Ex\left[z^{I_t}w^{R_t}e^{-s\Lambda(t)}\right]$ for $|z|\leq 1$ and $|w|\leq 1$. Then, $\varphi:=\varphi(t,z,w,s)$  satisfies the PDE
\begin{equation}\label{pde:Gen:Fun}
\begin{aligned}
\dfrac{\partial }{\partial t}\varphi+\left[rs-1+\left(\frac{1}{\alpha}\right)^sM_Q(s)z\right] \frac{\partial}{\partial s}\varphi&+\left[\left(\mu+rs\log\left(\alpha\right)\right)z-\mu\left(\frac{\alpha}{\beta}\right)^sw\right]\frac{\partial}{\partial z}\varphi\\[.3cm]
&+\bigg[\log(\beta)rsw\bigg]\frac{\partial}{\partial w}\varphi=
-\lambda_0rs
\varphi
\end{aligned}
\end{equation}
with initial condition
\[
\varphi(0,z,w,s)=z^{i_0}e^{-s(\lambda_0+i_0\log(\alpha))}
\]
\end{thm}

Since the resulting PDE is linear of first order, we can apply the method of characteristics to simplify the problem to a system of 2 by 2 ODEs. However, the solution can be quite messy and numerical methods are very likely to be needed to solve such system which we just point out for the sake of completeness.

\begin{cor} Let $(z,w,s)$ be a fixed point in $[0,\infty)^3$ and let $\Xi(t):=\varphi(t,z,w,s)=\Ex\left[z^{I_t}w^{R_t}e^{-s\Lambda(t)}\right]$. Then, $\Xi(t)$ is the solution of the system of ODEs

\small
\begin{equation}\label{Sys:ODE:Gen:Fun}
\begin{aligned}
\left\{\begin{array}{ccll}
\dfrac{d\vartheta}{dt}&=&-1+r\vartheta(t)+\left(\frac{1}{\alpha}\right)^{\vartheta(t)}M_Q(\vartheta(t))\bigg[&\hspace{-.25cm}z\exp\left(\mu t+r\log\left(\alpha\right)\displaystyle\int_0^t \vartheta(u)du\right)\\[.4cm]
&&&-w\mu\displaystyle\int_0^t\left(\frac{\alpha}{\beta}\right)^{\vartheta(u)}\exp\left(-\mu u+r\log(\beta)\displaystyle\int_0^u \vartheta(y)dy\right)du\bigg]\\[.4cm]
\dfrac{d\Xi}{dt}&=& -\lambda_0r\vartheta(t)\span \\[.4cm]
\vartheta(0)&=&s&\\[.5cm]
\Xi(0)&=&z^{i_0}\exp\left(-\vartheta(\lambda_0+i_0\log(\alpha))\right)&
\end{array}\right.
\end{aligned}
\end{equation}
\normalsize

\end{cor}

Our objective is to be able to compute the moments of the random variables $I_t$ and $R_t$, however the system of ODEs (\ref{Sys:ODE:Gen:Fun}) might not prove that useful for this task, therefore we turn our attention to the PDE (\ref{pde:Gen:Fun}) again. In order to obtain the joint $(i,j,k)-$th moment of $I_t$, $R_t$ and $\Lambda(t)$ we take the $i-$th derivative of the PDE (\ref{pde:Gen:Fun}) with respect to $z$, the $j-$th derivative with respect to $w$ and the $k-$th derivative with respect to $s$ plug in the values $z=w=1$ and $s=0$. However, for simplicity, we will proceed to give an explicit formula for the first and second moments of the processes $I_t,R_t,\Lambda(t)$ since this are the ones that might be used the most.

Unfortunately, the moments of $I_t$ and $R_t$ are not independent of those of $\Lambda(t)$, therefore to compute the first moment, we will need to solve a system of 3 linear differential equations. Each equation will be obtained by differentiating the PDE (\ref{pde:Gen:Fun}) with respect to one parameter $z,w$ or $s$ and plugging the values mentioned before. We show the procedure in the next two lemmas.

\begin{lemma} \label{lemma:firtsmom} For any time $t>0$, let $\Xt=\big[\Ex[\Lambda(t)],\Ex[I_t],\Ex[R_t]\big]^T$. Then, $\Xt$ is the solution to the system of differential equations
\begin{align}\label{system:ODE:firstmoment}
\dfrac{d}{dt}\Xt&=\mathbf{A}_1\Xt + \mathbf{b}_1\\[.4cm] \label{system:ODE:first:initial}
\overline{\mathbf{X}}_0&=\left[\begin{array}{c}\lambda_0+i_0\log(\alpha)\\i_0\\0\end{array}\right]
\end{align}
where   
\[
\mathbf{A}_1=\left[\begin{array}{ccc} \log\left(\alpha\right)+\Ex[Q]-r & (r-\mu)\log\left(\alpha\right)+\mu\log\left(\beta\right)& r\log(\beta)\\ 1 &-\mu&0\\ 0&\mu&0\end{array}\right]\qquad\text{and}\qquad\mathbf{b}_1=\left[\begin{array}{c} \lambda_0r \\ 0\\ 0\end{array}\right]
\]
That is,
\begin{equation}\label{firstmoments}
\Xt=\overline{\mathbf{X}}_0e^{\mathbf{A}_1 t} + \int_0^te^{\mathbf{A}_1(t-s)}\mathbf{b}_1ds
\end{equation}
\end{lemma}

We can use the same approach as in Lemma \ref{lemma:firtsmom} to compute the second moments. Besides, we will need to consider all 6 possible double products of the random variables $I_t,R_t$ and $\Lambda (t)$.

For the following Lemma, it will be useful to simplify the notation so that the matrix used to compute second moments is displayed nicely. Indeed, let
\begin{align*}
C_1&:=\log\left(\alpha\right)+\Ex[Q]-r\\
C_2&:=(r-\mu)\log\left(\alpha\right)+\mu\log(\beta)\\
C_3&:=r\log(\beta)\\
C_4&:=r\log(\alpha)
\end{align*}

An important remark is that without too much work we can transition from the models presented in Section \ref{Sec:Model} by just modifying the parameters of the original model. For example, if a pure-Hawkes (pH) model with intensity given by Equation \eqref{eqn:purehawkes} wants to be considered; then by setting $\alpha=\beta=1$ in the previous result it is obtained that the following:
\begin{align}\label{system:ODE:firstmoment:purehawkes}
\dfrac{d}{dt}\Xt^{pH}&=\mathbf{A}_1^{pH}\Xt + \mathbf{b}_1^{pH}\\[.4cm] \label{system:ODE:purehawkes:initial}
\overline{\mathbf{X}}_0^{pH}&=\left[\begin{array}{c}\lambda_0\\i_0\\0\end{array}\right]
\end{align}
where
\[
\mathbf{A}_1^{pH}=\left[\begin{array}{ccc} \Ex[Q]-r & 0& 0\\ 1 &-\mu&0\\ 0&\mu&0\end{array}\right]\qquad\text{and}\qquad\mathbf{b}_1=\left[\begin{array}{c} \lambda_0r \\ 0\\ 0\end{array}\right]
\]
which agrees with the result provided in \cite{koops2018}. Further, if a model whose infection rate increases or decreases proportional to the number of infected people in its baseline intensity (pI) as given by Equation \eqref{eqn:proptoI} is sought, then it should be set $\beta=1$ and the resulting system of ODEs is

\begin{align}\label{system:ODE:pI}
\dfrac{d}{dt}\Xt^{pI}&=\mathbf{A}_1\Xt^{pI} + \mathbf{b}_1^{pI}\\[.4cm] \label{system:ODE:pI:initial}
\overline{\mathbf{X}}_0^{pI}&=\left[\begin{array}{c}\lambda_0+i_0\log(\alpha)\\i_0\\0\end{array}\right]
\end{align}
where
\[
\mathbf{A}_1^{pI}=\left[\begin{array}{ccc} \log\left(\alpha\right)+\Ex[Q]-r & (r-\mu)\log\left(\alpha\right)& 0\\ 1 &-\mu&0\\ 0&\mu&0\end{array}\right]\qquad\text{and}\qquad\mathbf{b}_1=\left[\begin{array}{c} \lambda_0r \\ 0\\ 0\end{array}\right]
\]

\begin{rem}
Notice that in the two particular cases above given by systems of ODEs in Equations \eqref{system:ODE:firstmoment:purehawkes}-\eqref{system:ODE:pI:initial} since $\beta=1$, then the intensity becomes independent of $R_t$ and thus the entire system becomes independent of $R_t$ which is reflected by the third column of the matrices $\mathbf{A}_1^{pH}$ and $\mathbf{A}_1^{pI}$ being 0.
\end{rem}

Finally, a third important differentiation from the base case $\alpha,\beta>1$ is when the number of recovered individuals actually \emph{decreases} the likelihood of a new infection, such as in a classical SIR model and which is also described  in Section \ref{Sec:Model}. In this case, and as mentioned earlier, provided that the process is stopped whenever the intensity becomes negative, the corresponding system of differential equations for the first moments can be obtained from Lemma \ref{lemma:firtsmom} by setting $\beta=1/\tilde{\beta}$ with $\tilde{\beta}>1$ so that, in this case, $\beta<1$ and the baseline intensity given by Assumption \eqref{assump:baseline} becomes
\begin{equation}\label{eqn:intensity:SIR}
 \lambda^\infty(i,\rho)=\lambda_0+i\log(\alpha)+\rho\log(\beta) = \lambda_0+i\log(\alpha)-\rho\log(\tilde{\beta}).
\end{equation}
In this case, the corresponding system of differential equations governing the first moments of this system is given by

\begin{align}\label{system:ODE:SIR}
\dfrac{d}{dt}\Xt^{SIR}&=\mathbf{A}_1\Xt^{SIR} + \mathbf{b}_1^{SIR}\\[.4cm] \label{system:ODE:first:initial}
\overline{\mathbf{X}}_0^{SIR}&=\left[\begin{array}{c}\lambda_0+i_0\log(\alpha)\\i_0\\0\end{array}\right]
\end{align}
where
\[
\mathbf{A}_1^{SIR}=\left[\begin{array}{ccc} \log\left(\alpha\right)+\Ex[Q]-r & (r-\mu)\log\left(\alpha\right)-\mu\log(\tilde{\beta})& -r\log(\tilde{\beta})\\ 1 &-\mu&0\\ 0&\mu&0\end{array}\right]\qquad\text{and}\qquad\mathbf{b}_1=\left[\begin{array}{c} \lambda_0r \\ 0\\ 0\end{array}\right]
\]

Next, a characterization of the second moments is provided. The techniques and methods to obtain them follow from the ones used in Lemma \ref{lemma:firtsmom} but to solve this system it is necessary to obtain the solution of the system \eqref{system:ODE:firstmoment}-\eqref{system:ODE:first:initial}, since its solution its dependent (as expected) on the first moments of $(\Lambda(t),I_t,R_t)$.

\begin{lemma}\label{lemma:secondmom} For any time $t>0$, let $\Yt=\Big[\Ex[\Lambda^2(t),\Ex[I_t(I_t-1)],\Ex[R_t(R_t-1)], \Ex[\Lambda(t)I_t],\Ex[\Lambda(t)R_t],\Ex[I_tR_t]\Big]^T$. Then, $\Yt$ is the solution to the system of equations
\begin{align}\label{system:ODE:secondmoment}
\dfrac{d}{dt}\Yt&=\mathbf{A}_2\Yt + \mathbf{b}_2\\[.4cm] \label{system:ODE:second:initial}
\overline{\mathbf{Y}}_0&=\left[(\lambda_0+i_0\log(\alpha)))^2,i_0^2,0,i_0(\lambda_0+i_0\log(\alpha)),0,0\right]^T
\end{align}
where,

\[
\mathbf{A}_2=\left[\begin{array}{cccccc} 
2C_1&    0   &  0  &2C_2  & 2C_3& 0\\ 
 0  & -2\mu  &  0  &  2   &  0  & 0 \\
 0  &   0    &  0  &  0   &  0  & 2\mu \\  
 1  &   C_2  &  0  & C_1  &  0  & C_3 \\
 0  &    0   & C_3 & -\mu & C_1 & C_2\\
 0  &   \mu  &  0  &  0   &  1  & -\mu \\
\end{array}\right]
\]
\normalsize
and
\[
\mathbf{b}_2=\left[\begin{array}{c} 
\left[\log^2\left(\alpha\right)+2\log(\alpha)+\Ex[Q^2]+2\lambda_0r\right]\Ex[\Lambda(t)] + \mu\log^2\left(\frac{\alpha}{\beta}\right)\Ex[I_t]\\ 
0 \\ 
0 \\ 
(C_1+r)\Ex[\Lambda(t)]+r\lambda_0\Ex[I_t]+C_4\Ex[R_t] \\ 
(r\lambda_0+C_3)\Ex[R_t]+(C_2-C_4)\Ex[I_t]\\
 0
\end{array}\right]
\]

That is,
\begin{equation}\label{secondmoments}
\Yt=\overline{\mathbf{Y}}_0e^{\mathbf{A}_2 t} + \int_0^te^{\mathbf{A}_2(t-s)}\mathbf{b}_2ds
\end{equation}
\end{lemma}

\section{Some considerations regarding the Estimation of the State Dependent Hawkes Process} \label{section:estimation}

As with many many point processes, there are mainly three types of estimation procedures: MLE, (generalized)MoM and LSE-type methods. Below, a discussion on each of these three process, their advantages and disadvantages is presented. It is important to emphasize that there is not a definite answer into which method provides an advantage when doing inference over the parameters of the model and this intricate topic is left out for future research. 

\begin{itemize}
	\item \textbf{Maximum Likelihood Estimation (MLE)}. This method is one of the most used and regarded within the academic literature. This is due in part because there there are many theoretical results that guarantee that the MLE methods will yield an optimal solution. Unfortunately, for many Markovian point processes the evaluation of the log-likelihood function is of order $O(N_t^2)$. Several algorithms such as EM are used to improve this (see \cite{lewis2012self}), but in general, the biggest problem of MLE is that the likelihood curve is very flat, with many local maxima and without a clear way to decide which is the global maximum. An important step in improving this is discussed in \cite{lewis2012self} and \cite{veen2008estimation}. Many of these methods apply to the exponential kernel form used in this paper (see Assumption \eqref{assump:phi}), but also have a constant baseline intensity. It is still an ongoing research topic of this author to generalize such methods to a state dependent intensity. 
	\item \textbf{(Generalized) Method of  Moments (gMoM)}. These methods are derived under the assumption that the process is in its limiting stationary state, and as any Method of Moments, the idea is to form a system of equations of $n\times n$, where $n$ is the number of parameters to be determined and $n$ linear independent equations relating the moments are used. However, for Hawkes processes, usually there are not enough linear independent conditions on the moments and thus the Autocovariance function is introduced to generate other equations. A prime example is given in the work by \cite{foschi2020warnings}. Further, these methods are inefficient when the number of observations is ``small'' and they usually does not work properly on higher dimensions. However, it is important to mention that in our case, these methods might not be the best since we are trying to model the early stages of the epidemic where the processes are far away from the stationary state. 
	\item \textbf{Least Square Estimation (LSE) type methods}. These methods have not been explored until recently, mainly because the order of these methods is similar to the ones of MLE. However, the recent work by \cite{cartea2021gradient}  show that unlike MLE methods, LSE methods can possess certain algebraic properties that help with the stochastic approximation of the kernels to then maximize or minimize the LSE functional. 
\end{itemize}

The nature of the process we are dealing with requires special care in the estimation of parameters and as such, this paper will not try to just follow the MLE method presented in Daley Vere Jones, Section , but rather this sensitive topic is left as an object of further research. Specially if Assumption \eqref{assump:phi} is relaxed and we allow for a more general kernel and also a state-dependent self-exciting kernel.

\section{Empirical and Numerical Examples}\label{Sec:Numerical} 

The objective of this section is to provide numerical evidence of the behaviour that the model has under the different scenarios proposed as well as the implications of increasing or decreasing the level of quarantine provided by the random variable $Q$ and other change in the parameters.

For all the experiments there will be four levels of Quarantine: high ($Q_H$), medium-high ($Q_{M_H}$), medium-low ($Q_{M_L}$) and low ($Q_L$). To completely specify the levels of quarantine provided by these random variables, their distribution is presented next.
\small
\[
Q_H=\begin{cases} 0.05 & \text{w.p. }0.92\\ 0.45 & \text{w.p. }0.03\\ 0.60 & \text{w.p. }0.02\\ 0.95 & \text{w.p. }0.03 \end{cases},\quad  Q_{M_H}=\begin{cases} 0.05 & \text{w.p. }0.11\\ 0.45 & \text{w.p. }0.18\\ 0.60 & \text{w.p. }0.46\\ 0.95 & \text{w.p. }0.25 \end{cases},\quad  Q_{M_L}=\begin{cases} 0.05 & \text{w.p. }0.12\\ 0.45 & \text{w.p. }0.46\\ 0.60 & \text{w.p. }0.32\\ 0.95 & \text{w.p. }0.10 \end{cases},\quad Q_L=\begin{cases} 0.05 & \text{w.p. }0.02\\ 0.45 & \text{w.p. }0.02\\ 0.60 & \text{w.p. }0.02\\ 0.95 & \text{w.p. }0.91 \end{cases}
\] 
\normalsize
As it can be seen, the random variables take the same values in all cases but their probabilities change. They are specified so that $\Ex[Q_L]=0.9$, $\Ex[Q_{M_L}]=0.6$, $\Ex[Q_{M_H}]=0.5$ and $\Ex[Q_H]=0.1$.

Next, we provide several figures with various simulations of the different cases mentioned in Section \ref{Sec:Quantities}. In particular, it is important to notice how the difference on the parameters affect the speed at which the number of infected individuals grow. Also, it is important to remark that although the simulation algorithm \ref{algo:SIR} is a variation of the thinning algorithm by Ogata described in Appendix \ref{Sec:algorithms}, to the author's knowledge, an explicit algorithm to simulate a Hawkes process with a state dependent intensity is not readily available. Thus, as part of this work, a detailed pseudo-code for simulating these kind of processes is provided in the Appendix \ref{Sec:algorithms}. Furthermore, due to the immense amount of simulations and computations required to simulate these processes -at the end the thinning algorithm is a variant of an acceptance-rejection method and as such, many simulations are rejected- the process cannot be easily simulated for large time intervals.

To understand the behaviour of the process counting the number of infected people some different parameters of the model are changed, in Figures \ref{fig:pure_hawkes}-\ref{fig:beta_less_1} below, the number of infected people is simulated for the small time interval $[0,4]$ together with the corresponding solution of the differential equation under different scenarios, where either $\alpha$, $\beta$ or the Quarantine distribution changes. This exercise has two purposes: verify that the simulation algorithm and the differential equations yield similar results (performing a cross-verification) and showing how the change in different parameters yield logical conclusions as well as exploring how such changes affect the speed at which the number of infected people grows. 

To understand the behaviour of the process counting the number of infected people under a general purely Hawkes process (i.e. we set $\alpha=\beta=1$), in Figure \ref{fig:pure_hawkes}, $I_t$ is plotted under 12 different scenarios and each scenario under the 4 different Quarantine scenarios. In all the scenarios, the parameters were set to be the same except for the decaying parameter $r$ (see Equation \eqref{assump:phi}) and the constant baseline intensity $\lambda_0$. As expected, the number of infected people decreases as $r$ gets larger, meaning that an infectious person will contribute to a new infection significantly only during a short time span. 

Two interesting cases arrive in the general state dependent Hawkes baseline intensity case. First, as can be seen in Figure  \ref{fig:alpha_more_beta}, when $\alpha>\beta$ (in this case $\alpha=1.4$ while $\beta=1.2$, we see that there are more infections as compared to the purely Hawkes case reflected in Figure \ref{fig:pure_hawkes}. However, when $\beta>\alpha>1$, we see that the infections also grow as compared within the baseline case of the purely Hawkes model but they are actually even higher than in the previous case where $\alpha>\beta$. The reason for this is that the number of infected people becomes, as time progresses, comparative smaller to the number of people that have recovered and continue propagating the virus. This is a classical behaviour of viruses that do not create immunity as is the case of COVID-19. Finally, in Figure \ref{fig:beta_less_1} we plot the case where $\alpha>1$ but $\beta<1$. This is to illustrate how a $\beta$ less than 1 can mitigate greatly the size of the epidemics. Indeed, as a good comparison, note that, with all the parameters save $\alpha$ and $\beta$ kept constant, the maximum range up to time 4 for the number of infected individuals in the Pure Hawkes case is of roughly 400, while on the state-dependent Hawkes cases with $\alpha>\beta>1$ is of 600; when $\beta>\alpha>1$ is of 1200 while on the case where $\beta<1<\alpha$ it is of 200.

Finally, in Figure \ref{fig:long_term}, the long-run behavior of the state-dependent Hawkes process with $\beta<1$ is displayed. However, as mentioned above the simulation of the process for larger time windows is not feasible and as such only the solution of the differential equation is provided. The purpose of this figure is to show that under this case, we can observe the typical behaviour of an SIR model. Since this was computed in the long run, we only plot $t$ vs $\Ex[I_t]$.

\begin{figure}[H]
	\centering
		\includegraphics[width=0.9\textwidth]{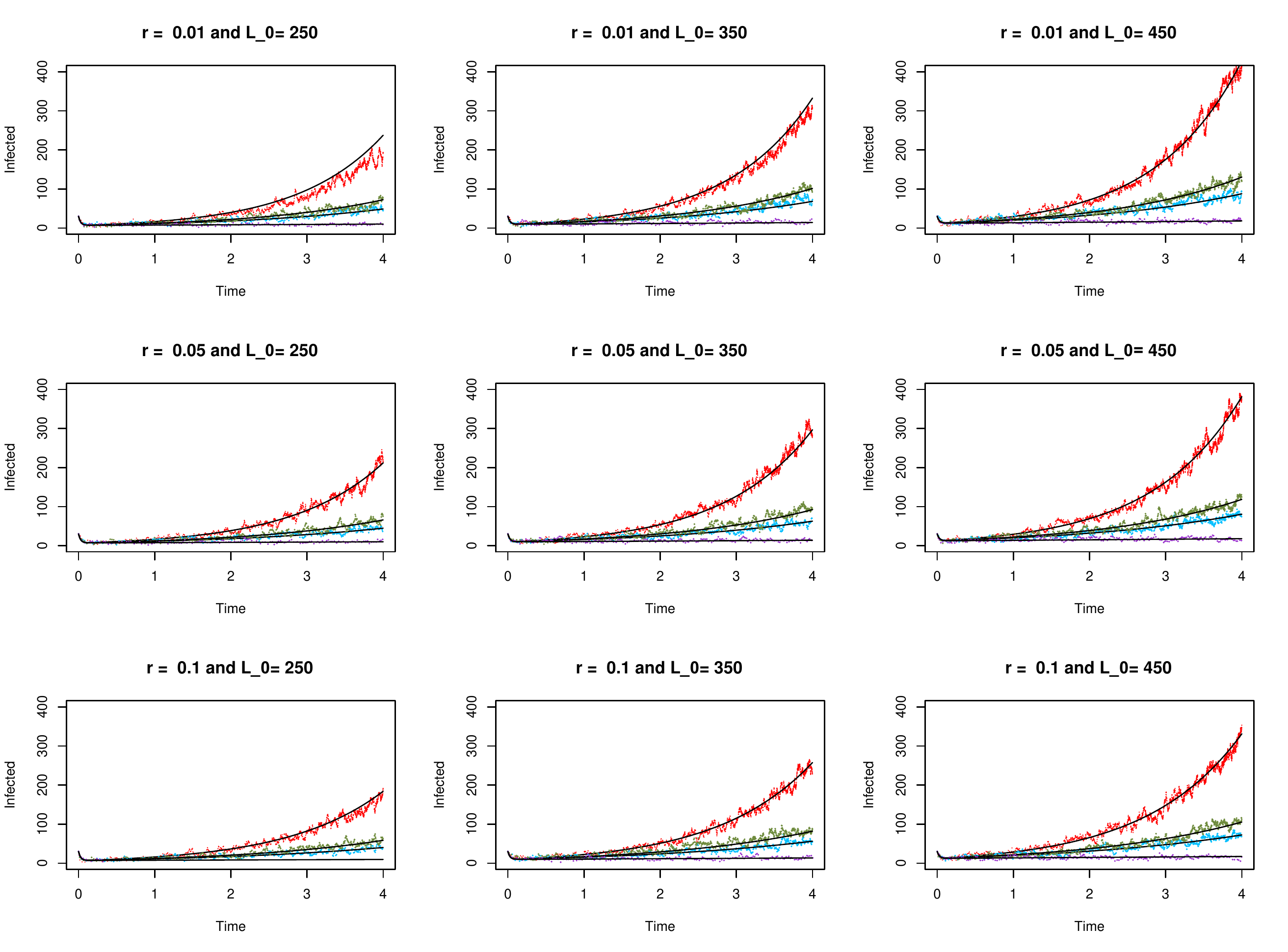}
	\caption{Figure showing 12 simulations and their behaviour for different parameters under a purely Hawkes approach. Plots in the same row share the same value of $r$ while plots in the same column share the same value of $\lambda_0$. The rest of the parameters, as well as the scale, where kept the same for comparison purposes. The Quarantine level is also represented, with High Quarantine ($Q_H$) in purple, Medium-High Quarantine ($Q_{M_H}$) in blue, Medium-Low Quarantine ($Q_{M_L}$) in green and Low Quarantine ($Q_L$) in red. Furthermore, the mean (deterministic) behaviour of the system is plotted with the same color as the simulated paths.}
	\label{fig:pure_hawkes}
\end{figure}

\begin{figure}[H]
	\centering
		\includegraphics[width=0.9\textwidth]{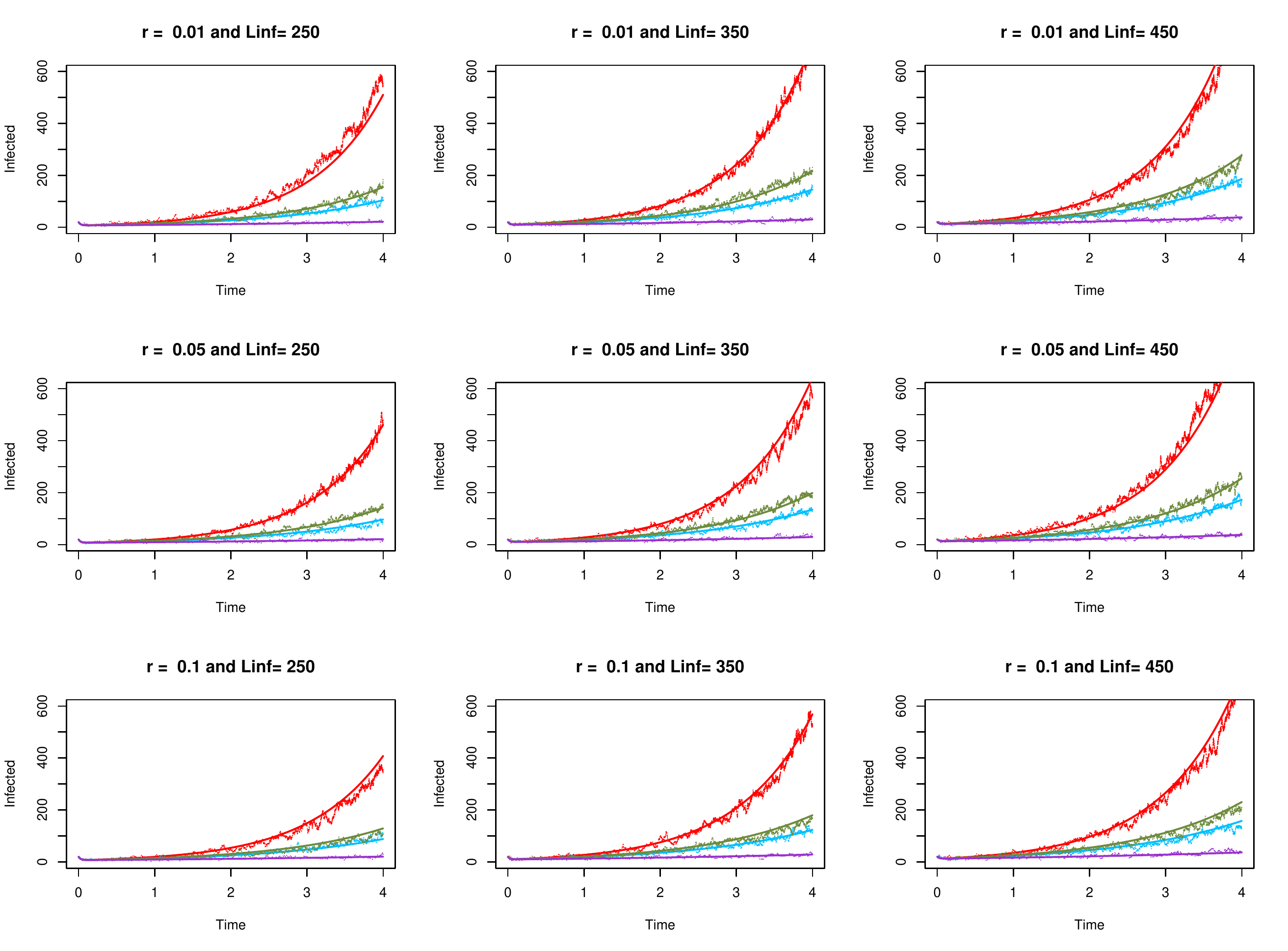}
	\caption{Figure showing 12 simulations and their behaviour for different parameters under a general state dependent Hawkes intensity with $\alpha=1.4$ and $\beta=1.2$. Plots in the same row share the same value of $r$ while plots in the same column share the same value of $\lambda_0$. The rest of the parameters, as well as the scale, where kept the same for comparison purposes. Furthermore, the Quarantine level is also represented, with High Quarantine ($Q_H$) in purple, Medium-High Quarantine ($Q_{M_H}$) in blue, Medium-Low Quarantine ($Q_{M_L}$) in green and Low Quarantine ($Q_L$) in red.}
	\label{fig:alpha_more_beta}
\end{figure}

\begin{figure}[H]
	\centering
		\includegraphics[width=0.9\textwidth]{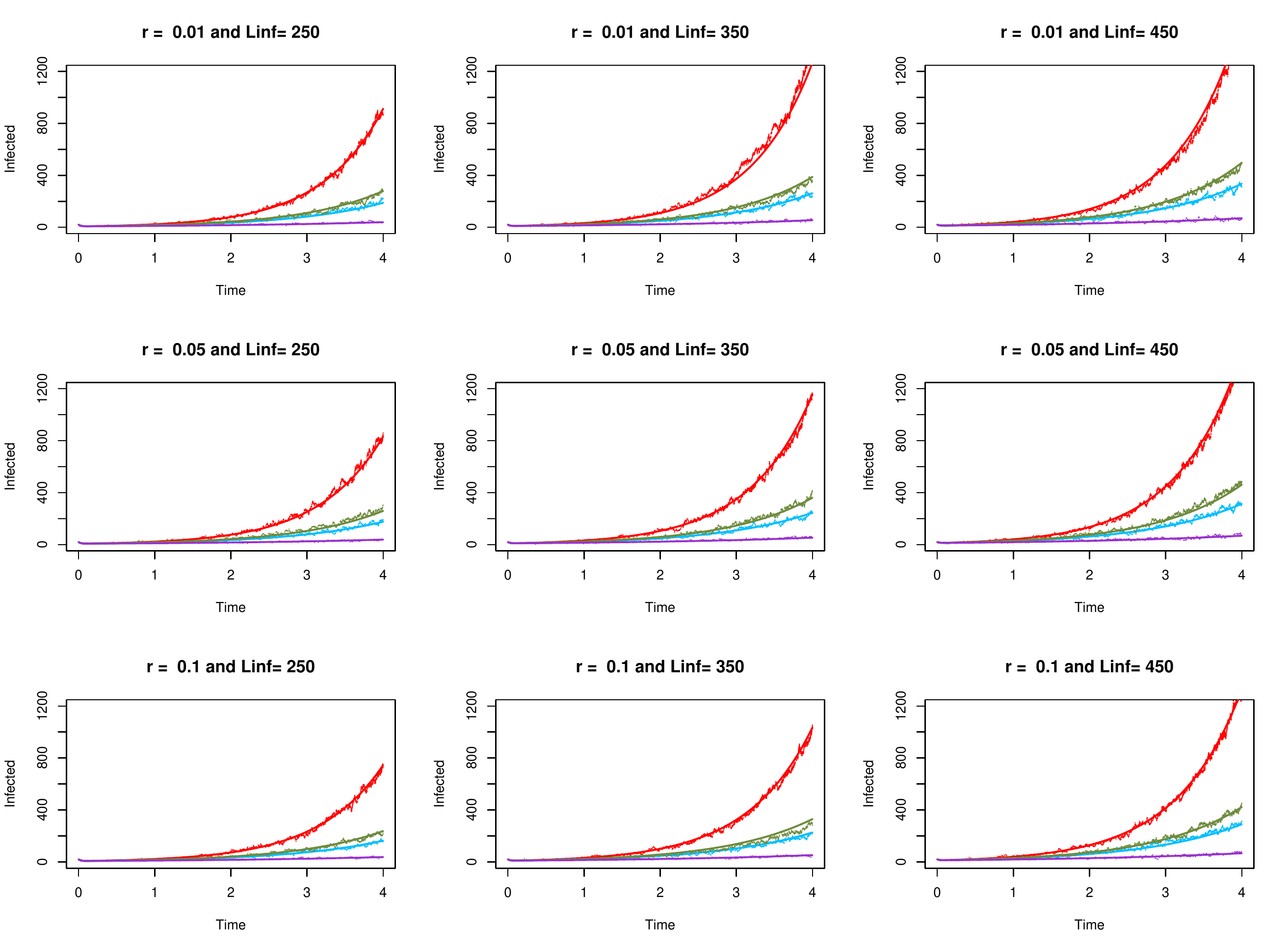}
	\caption{Figure showing 12 simulations and their behaviour for different parameters under a general state dependent Hawkes intensity with $\alpha=1.2$ and $\beta=1.4$. Plots in the same row share the same value of $r$ while plots in the same column share the same value of $\lambda_0$. The rest of the parameters, as well as the scale, where kept the same for comparison purposes. Furthermore, the Quarantine level is also represented, with High Quarantine ($Q_H$) in purple, Medium-High Quarantine ($Q_{M_H}$) in blue, Medium-Low Quarantine ($Q_{M_L}$) in green and Low Quarantine ($Q_L$) in red.}
	\label{fig:alpha_less_beta}
\end{figure}

\begin{figure}[H]
	\centering
		\includegraphics[width=0.9\textwidth]{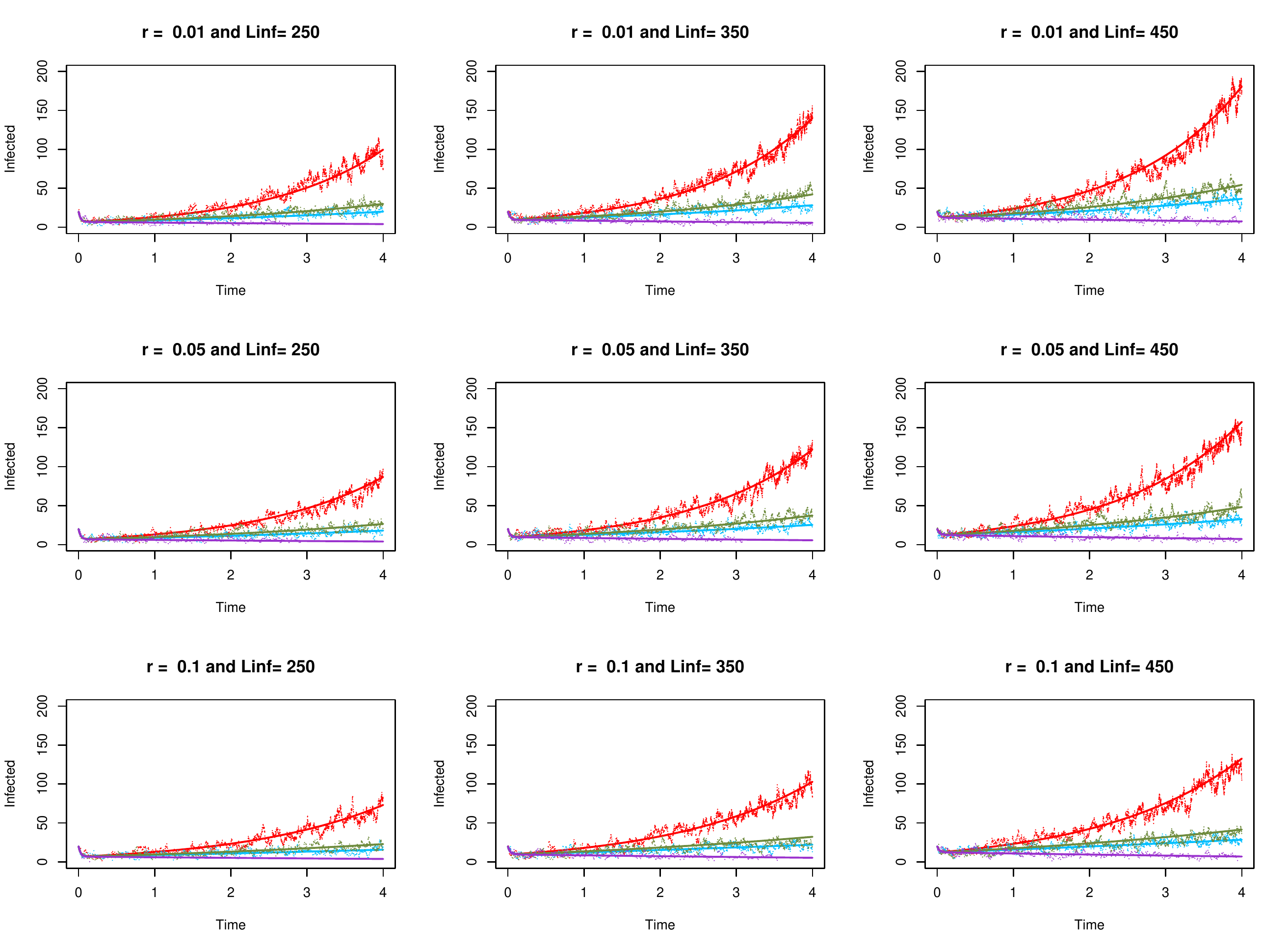}
	\caption{Figure showing 12 simulations and their behaviour for different parameters under a general state dependent Hawkes intensity with $\alpha=1.2$ and $\beta=0.8$. Plots in the same row share the same value of $r$ while plots in the same column share the same value of $\lambda_0$. The rest of the parameters, as well as the scale, where kept the same for comparison purposes. Furthermore, the Quarantine level is also represented, with High Quarantine ($Q_H$) in purple, Medium-High Quarantine ($Q_{M_H}$) in blue, Medium-Low Quarantine ($Q_{M_L}$) in green and Low Quarantine ($Q_L$) in red.}
	\label{fig:beta_less_1}
\end{figure}

\begin{figure}[H]
	\centering
		\includegraphics[width=0.6\textwidth]{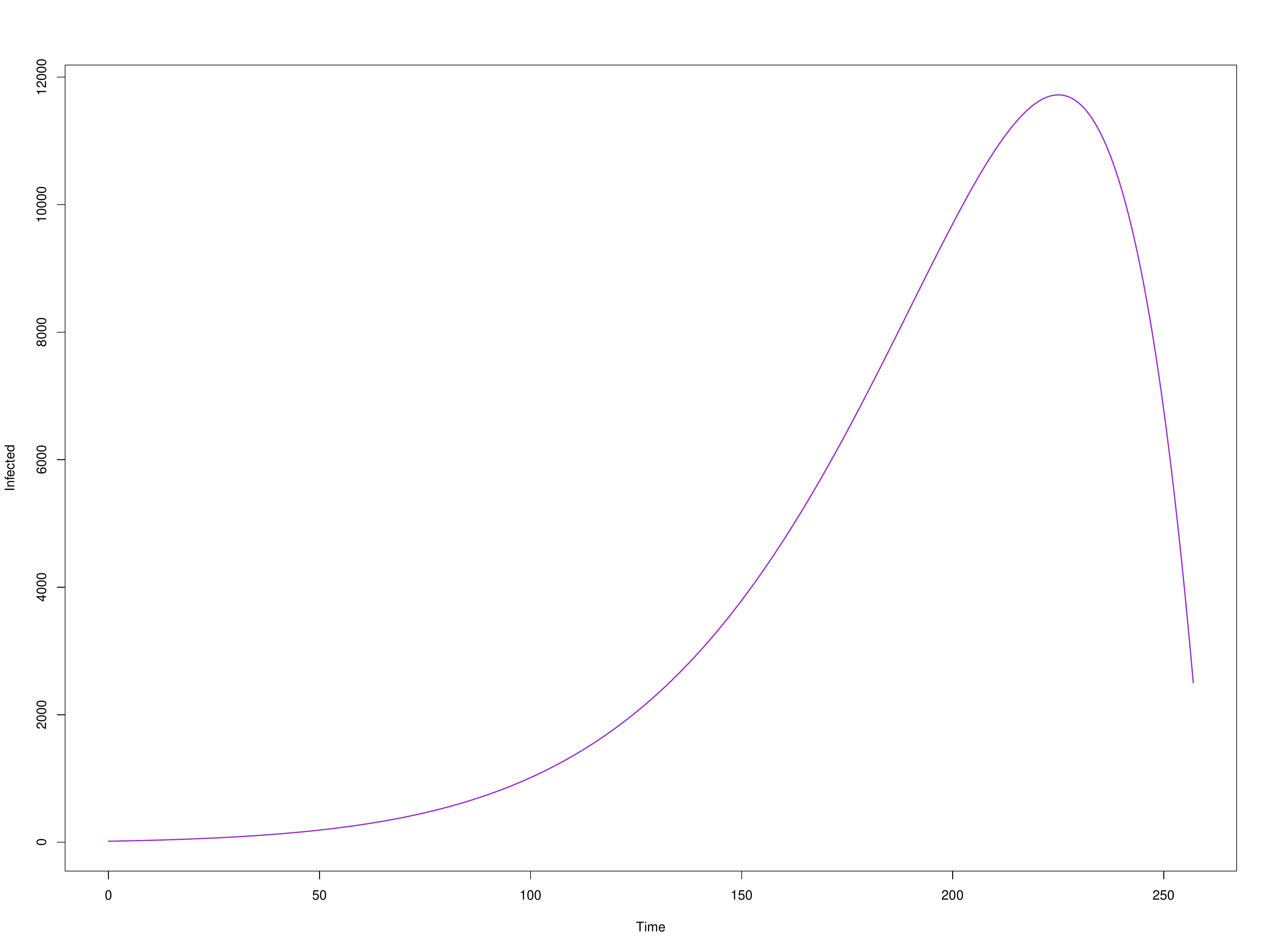}
	\caption{Figure showing the long term behaviour of the mean for the number of infected people under a general state dependent Hawkes intensity.}
	\label{fig:long_term}
\end{figure}

\section{Conclusions and further Research Direction}\label{Sec:Conclusions}

A phenomena that has been observed in the current COVID-19 epidemics is the overdispersion of process describing the number of contagions $C_t$. That is, that for different time windows, the variance of the number of new contagions within that time interval is great larger than their average. I.e., that for any real number $\tau>0$,
\[
\text{Var}[C_{t+\tau}-C_t]\gg\Ex[C_{t+\tau}-C_t]
\]

This work tries to construct a model that takes into account the overdispersion observed in epidemics with a high control reproduction number that the classic stochastic-SIR models fail to capture due to their Poissonian nature (where the expectation and variance are roughly the same for any time interval).

One possible solution to model the overdispersion is through the usage of regular point processes, in particular self-exciting point processes. Not only can they capture this phenomena but the rationale behind them makes sense. Every new infection will increase the likelihood of a new infection occurring. However, one of the main difficulties, and a very active are of research currently, is the efficient estimation of its parameters. While it is part of the ultimate goal of this work to research efficient estimation methods, it is not the primary intent of this paper. The basic estimation method is in terms of the likelihood function, but there is no clear way to determine whether the method found a local maxima or a global one. Furthermore, the likelihood curve is very flat and although theoretical convergence is guaranteed, in practice it is hard to achieve.  

To continue this research, there are two interesting directions to pursue. First to explore a regime-switching model where the parameters of the model change or switch randomly to different scenarios considering possible social, economical and political factors. Particularly of interest is the level of quarantining for the general population, where an stochastic optimization problem can be devised for minimizing the economic impact of a disease like COVID-19. Another direction of research is creating a full compartmental model where the transitions to different compartments are guided by point processes. This approach to classical stochastic compartmental models is of interest due to the flexibility and the statistical properties of point processes.

\appendix

\allowdisplaybreaks

\section{Proofs of Section \ref{Sec:Quantities}}\label{Sec:Proofs}

\begin{proof}[Proof of Theorem \ref{thm:pde:f}]
    Let $\Dt>0$ denote an infinitesimal time step. Due to the dynamics of the SIR model, from time $t$ to time $t+\Dt$, and since we assume that the point processes defining the arrivals of events is simple (see Equation \ref{def:simple}),only one of following four (disjoint) events can happen:
		
		\begin{itemize}
			\item[$\blacksquare$] Event of type (I): There is a new infection but no recoveries. That is,
			\[
			C_{t+\Dt}-C_t=1\qquad\qquad\text{and}\qquad\qquad R_{t+\Dt}-R_t=0
			\]
			\item[$\blacksquare$] Event of type (II): There is no new infection but a recovery. Then,
			\[
			C_{t+\Dt}-C_t=0\qquad\qquad\text{and}\qquad\qquad R_{t+\Dt}-R_t=1
			\]
			\item[$\blacksquare$] Event of type (III): There is a new infection and a recovery. Thus,
			\[
			C_{t+\Dt}-C_t=1\qquad\qquad\text{and}\qquad\qquad R_{t+\Dt}-R_t=1
			\]
			\item[$\blacksquare$] Event of type (IV): There are no new infections neither recoveries. Therefore,
			\[
			C_{t+\Dt}-C_t=0\qquad\qquad\text{and}\qquad\qquad R_{t+\Dt}-R_t=0
			\]
		\end{itemize}
 Recall that $\{t_i\}_{i=1}^\infty$ are the jump times of the counting process $(C_t)_{t\geq0}$. Then, define the function
\begin{equation*}
J(t)=\int_0^t Q_s e^{-r(t-s)}dC_s=\sum\limits_{t_i\leq t}Q_ie^{-r(t-t_i)}.
\end{equation*}
	Fix the terminal state at time $t+\dt$, and computing
	\begin{align*}
	F\Big(t+\dt,&i,\rho,\lambda-r(\lambda-\lambda^\infty(i,\rho))\dt\Big)=\Px\Big[I_{t+\dt}=i,R_{t+\dt}=\rho,\Lambda(t+\dt)\leq \lambda-r(\lambda-\linf)\dt\Big]\\[.3cm]
	&=\Px\left[I_{t+\dt}=i,R_{t+\dt}=\rho,\lambda^\infty(I_{t+\dt},R_{t+\dt})+\sum\limits_{t_i\leq t+\dt}Q_ie^{-r(t+\dt-t_i)}\leq \lambda-r(\lambda-\linf)\dt\right]\\[.3cm]
	&=\Px\Big[I_{t+\dt}=i,R_{t+\dt}=\rho,\lambda^\infty(i,\rho)+J(t+\dt)\leq \lambda(1-r\dt)+r\dt\linf\Big]\\[.3cm]
	&=\Px\Big[I_{t+\dt}=i,R_{t+\dt}=\rho,(1-r\dt)\lambda^\infty(i,\rho)+J(t+\dt)\leq \lambda(1-r\dt)\Big]\\[.3cm]
	&=\Px\Big[I_{t+\dt}=i,R_{t+\dt}=\rho,(1-r\dt)\big[\lambda^\infty(I_t,R_t)+J(t)\big]+ (1-r\dt)\big[\lambda^\infty(i,\rho)-\lambda^\infty(I_t,R_t)\big]\\
	&\qquad\qquad\qquad+ \big(J(t+\dt)-(1-r\dt)J(t)\big)\leq \lambda(1-r\dt)\Big]\\[.3cm]
	&=\Px\Big[I_{t+\dt}=i,R_{t+\dt}=\rho,(1-r\dt)\Lambda(t)+ (1-r\dt)\big[\lambda^\infty(i,\rho)-\lambda^\infty(I_t,R_t)\big]\\
	&\qquad\qquad\qquad+ \Big(\sum\limits_{t_i\leq t+\dt}\left\{(1-r\dt)Q_ie^{-r(t-t_i)}+o(\dt)\right\}-(1-r\dt)J(t)\Big)\leq \lambda(1-r\dt)\Big]\\[.3cm]
	&=\Px\left[I_{t+\dt}=i,R_{t+\dt}=\rho,\Lambda(t)+ \big[\lambda^\infty(i,\rho)-\lambda^\infty(I_t,R_t)\big]+\sum\limits_{t<t_i\leq t+\dt}Q_ie^{-r(t-t_i)}+o(\dt)\leq \lambda\right]\\
	&=\Px\left[I_{t+\dt}=i,R_{t+\dt}=\rho,\Lambda(t)\leq \lambda- \big[\lambda^\infty(i,\rho)-\lambda^\infty(I_t,R_t)\big]-\sum\limits_{t<t_i\leq t+\dt}Q_i+o(\dt)\right]\\
	\end{align*}
	Since we are later going to take a difference quotient, divide by $\dt$ and take the limit as $dt\to0$, for ease of notation, we will disregard all the terms of lower order after displaying them once and we will introduce the following notation
	\begin{align*}
	\Delta\lambda^{\infty,1}_{i,\rho}&=\lambda^\infty(i,\rho)-\lambda^\infty(i-1,\rho)\\
	\Delta\lambda^{\infty,2}_{i,\rho}&=\lambda^\infty(i,\rho)-\lambda^\infty(i+1,\rho-1)\\
	\Delta\lambda^{\infty,3}_{i,\rho}&=\lambda^\infty(i,\rho)-\lambda^\infty(i,\rho-1)\\
	\Delta\lambda^{\infty,4}_{i,\rho}&=\lambda^\infty(i,\rho)-\lambda^\infty(i,\rho)=0\\
	\end{align*}

Define
\[
	E(t+\dt;t,i,\rho,\lambda)=\left\{I_{t+\dt}=i,R_{t+\dt}=\rho,\Lambda(t)\leq \lambda- \big[\lambda^\infty(i,\rho)-\lambda^\infty(I_t,R_t)\big]-\sum\limits_{t<t_i\leq t+\dt}Q_i\right\}
\]
Then, by recalling that the point processes are simple and that $\{Q,Q_i\}_{i=1}^\infty$ is a sequence of i.i.d. random variables, we have that
\[
	E(t+\dt;t,i,\rho,\lambda)=\left\{I_{t+\dt}=i,R_{t+\dt}=\rho,\Lambda(t)\leq \lambda- \big[\lambda^\infty(i,\rho)-\lambda^\infty(I_t,R_t)\big]-Q\right\}
\]
and thus we can defined the events
	\begin{align*}
	E_I(t,i,\rho,\lambda)&=\left\{I_{t}=i-1,R_{t}=\rho,\Lambda(t)\leq \lambda- \Delta\lambda^\infty_1(i,\rho)-Q,\; C_{t+\dt}-C_t=1,R_{t+\dt}-R_t=0 \right\}\\[.3cm]
	E_{II}(t,i,\rho,\lambda)&=\bigg\{I_{t}=i+1,R_{t}=\rho-1,\Lambda(t)\leq \lambda- \Delta\lambda^\infty_2(i,\rho),\;C_{t+\dt}-C_t=0,R_{t+\dt}-R_t=1 \bigg\}\\[.3cm]
	E_{III}(t,i,\rho,\lambda)&=\left\{I_{t}=i,R_{t}=\rho-1,\Lambda(t)\leq \lambda- \Delta\lambda^\infty_3(i,\rho)-Q,\;C_{t+\dt}-C_t=1,R_{t+\dt}-R_t=1 \right\}\\[.3cm]
	E_{IV}(t,i,\rho,\lambda)&=\bigg\{I_{t}=i,R_{t}=\rho,\Lambda(t)\leq \lambda- \Delta\lambda^\infty_4(i,\rho),\;C_{t+\dt}-C_t=0,R_{t+\dt}-R_t=0\big] \bigg\}\\[.3cm]
	\end{align*}
	By considering the likelihood of the 4 types of events happening and the dynamics of the SIR model, it is clear that
	\begin{equation}\label{eqn:sum:events}
	\Px\Big[E(t+\dt;t,i,\rho,\lambda)\Big]=\Px\Big[E_I(t,i,\rho,\lambda)\Big]+\Px\Big[E_{II}(t,i,\rho,\lambda)\Big]+\Px\Big[E_{III}(t,i,\rho,\lambda)\Big]+\Px\Big[E_{IV}(t,i,\rho,\lambda)\Big]
	\end{equation}
	In the following, for exposition purposes,  we will compute each of the fours summands separately, and we will only consider operations up to the order $\dt$. Further,
	Then, it follows that

	\begin{align*}
		\Px\Big[E_I(t,i,\rho,\lambda)\Big]&=\int_0^{\lambda-\Delta\lambda^{\infty,1}_{i,\rho}} f(t,i-1,\rho,y)\Px[C_{t+\dt}-C_t=1]\Px[R_{t+\dt}-R_t=0]\Px\Big[Q\leq \lambda-\Delta\lambda^{\infty,1}_{i,\rho}-y\Big]dy\\
		&=\int_0^{\lambda-\Delta\lambda^{\infty,1}_{i,\rho}} f(t,i-1,\rho,y)\Big[y\dt+o(\dt)\Big]e^{-(i-1)\mu\dt}\Px\Big[Q\leq \lambda-\Delta\lambda^{\infty,1}_{i,\rho}-y\Big]dy\\
		&=\int_0^{\lambda-\Delta\lambda^{\infty,1}_{i,\rho}} f(t,i-1,\rho,y)\Big[y\dt+o(\dt)\Big]\Big[1-(i-1)\mu\dt+o(\dt)\Big]\Px\Big[Q\leq \lambda-\Delta\lambda^{\infty,1}_{i,\rho}-y\Big]dy\\
		&=\int_0^{\lambda-\Delta\lambda^{\infty,1}_{i,\rho}} y\dt f(t,i-1,\rho,y)\Px\Big[Q\leq \lambda-\Delta\lambda^{\infty,1}_{i,\rho}-y\Big]dy+o(\dt),
	\end{align*}
	\begin{align*}
		\Px\Big[E_{II}(t,i,\rho,\lambda)\Big]&=\int_0^{\lambda-\Delta\lambda^{\infty,2}_{i,\rho}} f(t,i+1,\rho-1,y)\Px[C_{t+\dt}-C_t=0]\Px[R_{t+\dt}-R_t=1]dy\\
		&=\int_0^{\lambda-\Delta\lambda^{\infty,2}_{i,\rho}} f(t,i+1,\rho-1,y)\Big[1-y\dt + o(\dt)\Big]\left[1-e^{(i+1)\mu\dt}\right]dy\\
		&=\int_0^{\lambda-\Delta\lambda^{\infty,2}_{i,\rho}} f(t,i+1,\rho-1,y)\Big[1-y\dt + o(\dt)\Big]\left[1-(1-(i+1)\mu\dt+o(\dt))\right]dy\\
		&=\int_0^{\lambda-\Delta\lambda^{\infty,2}_{i,\rho}} f(t,i+1,\rho-1,y)\Big[1-y\dt + o(\dt)\Big](i+1)\mu\dt dy+o(\dt)\\
		&=\int_0^{\lambda-\Delta\lambda^{\infty,2}_{i,\rho}} (i+1)\mu\dt f(t,i+1,\rho-1,y)dy+o(\dt),
	\end{align*}
	
	\begin{align*}
		\Px\Big[E_{III}(t,i,\rho,\lambda)\Big]&=\int_0^{\lambda-\Delta\lambda^{\infty,3}_{i,\rho}} f(t,i,\rho-1,y)\Px[C_{t+\dt}-C_t=1]\Px[R_{t+\dt}-R_t=1]\Px\Big[Q\leq \lambda-\Delta\lambda^{\infty,3}_{i,\rho}-y\Big]dy\\
		&=\int_0^{\lambda-\Delta\lambda^{\infty,3}_{i,\rho}} f(t,i,\rho-1,y)\Big[y\dt+o(\dt)\Big]\Big[1-e^{-i\mu\dt}\Big]\Px\Big[Q\leq \lambda-\Delta\lambda^{\infty,3}_{i,\rho}-y\Big]dy\\
		&=\int_0^{\lambda-\Delta\lambda^{\infty,3}_{i,\rho}} f(t,i,\rho-1,y)\Big[y\dt+o(\dt)\Big]\Big[i\mu\dt+o(\dt)\Big]\Px\Big[Q\leq \lambda-\Delta\lambda^{\infty,3}_{i,\rho}-y\Big]dy\\
		&=\int_0^{\lambda-\Delta\lambda^{\infty,3}_{i,\rho}} yi\mu(\dt)^2f(t,i,\rho-1,y)\Px\Big[Q\leq \lambda-\Delta\lambda^{\infty,3}_{i,\rho}-y\Big]dy+o(\dt)\\
		&=o(\dt)
	\end{align*}
and
\begin{align*}
		\Px\Big[E_{IV}(t,i,\rho,\lambda)\Big]&=\int_0^{\lambda-\Delta\lambda^{\infty,4}_{i,\rho}} f(t,i,\rho,y)\Px[C_{t+\dt}-C_t=0]\Px[R_{t+\dt}-R_t=0]dy\\
		&=\int_0^{\lambda} f(t,i,\rho,y)\Big[1-y\dt + o(\dt)\Big]e^{i\mu\dt}dy\\
		&=\int_0^{\lambda} f(t,i+1,\rho-1,y)\Big[1-y\dt + o(\dt)\Big]\Big[1-i\mu\dt+o(\dt)\Big]dy\\
		&=(1-i\mu\dt)\int_0^{\lambda} f(t,i+1,\rho-1,y)\Big[1-y\dt + o(\dt)\Big]dy+o(\dt)\\
		&=(1-i\mu\dt)F(t,i,\rho,y)-\int_0^\lambda y\dt f(t,i,\rho,y)dy+o(\dt)\\
	\end{align*}
	
	Therefore, by using Equation (\ref{eqn:sum:events}) and recalling that $F\Big(t+\dt,i,\rho,\lambda-r(\lambda-\lambda^\infty(i,\rho))\dt\Big)=\Px\Big[E(t+\dt;t,i,\rho,\lambda)\Big]$, we have that
	\begin{align*}
	F\Big(t+\dt,i,\rho,\lambda-r(\lambda-\lambda^\infty(i,\rho))\dt\Big)&=\int_0^{\lambda-\Delta\lambda^{\infty,1}_{i,\rho}} y\dt f(t,i-1,\rho,y)\Px\Big[Q\leq \lambda-\Delta\lambda^{\infty,1}_{i,\rho}-y\Big]dy\\
	&+\int_0^{\lambda-\Delta\lambda^{\infty,2}_{i,\rho}} (i+1)\mu\dt f(t,i+1,\rho-1,y)dy\\
	&+(1-i\mu\dt)F(t,i,\rho,y)-\int_0^\lambda y\dt f(t,i,\rho,y)dy+o(\dt)
	\end{align*}
	Dividing by $\dt$ and letting $\dt\to0$,
		\begin{align*}
	\frac{\partial}{\partial t}F\Big(t,i,\rho,\lambda\Big)-r(\lambda-\lambda^\infty(i,\rho))\frac{\partial}{\partial \lambda}F(t,i,\rho,\lambda) &=\int_0^{\lambda-\Delta\lambda^{\infty,1}_{i,\rho}} y f(t,i-1,\rho,y)\Px\Big[Q\leq \lambda-\Delta\lambda^{\infty,1}_{i,\rho}-y\Big]dy\\
	&+\int_0^{\lambda-\Delta\lambda^{\infty,2}_{i,\rho}} (i+1)\mu f(t,i+1,\rho-1,y)dy\\
	&-i\mu F(t,i,\rho,y)-\int_0^\lambda y f(t,i,\rho,y)dy
	\end{align*}
	Differentiating with respect to $\lambda$; applying Leibniz Integral Rule where appropriate; and using Assumption (\ref{assump:Q}), we get the result.
\end{proof}

\begin{proof}[Proof of Theorem \ref{thm:pde1}]

Define
\begin{align} \label{defn:psi}
\psi(t,i,\rho,s)&:=\int_0^\infty e^{-s\lambda}f(t,i,r,\lambda)d\lambda\\ \label{defn:xi}
\varphi(t,z,w,s)&:=\Ex\left[z^{I_t}w^{R_t}e^{-s\Lambda(t)}\right].
\end{align}

 By definition,
\begin{align*}
\varphi(t,z,w,s)&=\Ex\left[z^{I_t}w^{R_t}e^{-s\Lambda(t)}\right]\\
		&=\sum\limits_{i=0}^\infty \sum\limits_{\rho=0}^\infty \int_0^\infty z^iw^\rho e^{-s\lambda} f(t,i,\rho,\lambda)d\lambda\\
		&=\sum\limits_{i=0}^\infty \sum\limits_{\rho=0}^\infty z^iw^\rho \int_0^\infty e^{-s\lambda} f(t,i,\rho,\lambda)d\lambda\\
		&=\sum\limits_{i=0}^\infty \sum\limits_{\rho=0}^\infty z^iw^\rho \psi(t,i,\rho,s)\\
\end{align*}
Thus, the strategy is to first use Theorem \ref{thm:pde:f} to find a PDDE equation that characterizes the function $\psi(\sbt)$ and then transform that PDDE as the corresponding one for $\varphi(\sbt)$.

Multiplying equation \ref{pde:f} by $e^{-s\lambda}$ and integrating over $(0,\infty)$ we get:
\begin{align*}
\dfrac{\partial }{\partial t}\int_0^\infty f(t,i,\rho,\lambda)e^{s\lambda}d\lambda&-\int_0^\infty\frac{\partial}{\partial \lambda}(r\lambda f(t,i,\rho,\lambda))e^{-s\lambda}d\lambda+r\lambda^\infty(i,\rho)\int_0^\infty e^{-s\lambda}\frac{\partial}{\partial \lambda} f(t,i,\rho,\lambda)d\lambda=\\
&\int_0^\infty \int_0^{\lambda-\Dl{1}}yf(t,i-1,\rho,y)d\Px[Q\leq \lambda-\Dl{1}-y]dye^{-s\lambda}d\lambda+\\
 & (i+1)\mu \int_0^\infty f(t,i+1,\rho-1,\lambda-\Dl{2})e^{s\lambda}d\lambda-\\
&\int_0^\infty(i\mu+\lambda)f(t,i,\rho,\lambda)e^{-s\lambda}d\lambda
\end{align*}
Using integration by parts and Fubini's theorem, and relation (\ref{defn:psi})
\begin{align*}
\dfrac{\partial }{\partial t}\psi(t,i,\rho,s)&-\int_0^\infty rs\lambda f(t,i,\rho,\lambda)e^{-s\lambda}d\lambda+r\lambda^\infty(i,\rho)\int_0^\infty s f(t,i,\rho,\lambda)e^{-s\lambda}d\lambda=\\
&\int_0^\infty \int_{y+\Dl1}^{\infty}yf(t,i-1,\rho,y)d\Px[Q\leq \lambda-\Dl{1}-y]e^{-s\lambda}d\lambda dy+\\
 & (i+1)\mu e^{-s\Dl2} \int_{-\Dl2}^\infty f(t,i+1,\rho-1,\lambda)e^{s\lambda}d\lambda-\\
&i\mu\psi(t,i,\rho,s)-\\
&\int_0^\infty\lambda f(t,i,\rho,\lambda)e^{-s\lambda}d\lambda
\end{align*}
 Let $M_Q(s)=\Ex[e^{-sQ}]$ denote the moment generating function of $Q$, which by assumption it exists on a neighbourhood of 0. Then, by using the fact that $f(t,i,\rho,\lambda)=0$ if $\lambda<0$, definition (\ref{defn:psi}) and Fubini's Theorem we can further simplify the above equation as
\begin{align*}
\dfrac{\partial }{\partial t}\psi(t,i,\rho,s)&+rs \frac{\partial}{\partial s}\psi(t,i,\rho,s)+rs\lambda^\infty(i,\rho)\psi(t,i,\rho,s)=\\
& e^{-s\Dl1}M_Q(s)\frac{\partial}{\partial s}\psi(t,i-1,\rho,s)+(i+1)\mu e^{-s\Dl2} \psi(t,i+1,\rho-1,s)\\
&-i\mu\psi(t,i,\rho,s)+\frac{\partial}{\partial s}\psi(t,i,\rho,s)
\end{align*}
Finally, rearranging terms we have
\begin{equation}\label{pde:f:2}
\begin{aligned}
\dfrac{\partial }{\partial t}\psi(t,i,\rho,s)&+(rs-1) \frac{\partial}{\partial s}\psi(t,i,\rho,s)+rs\lambda^\infty(i,\rho)\psi(t,i,\rho,s)+e^{-s\Dl1}M_Q(s)\frac{\partial}{\partial s}\psi(t,i-1,\rho,s)=\\[.3cm]
&(i+1)\mu e^{-s\Dl2} \psi(t,i+1,\rho-1,s)-i\mu\psi(t,i,\rho,s)
\end{aligned}
\end{equation}
At this point, we use the particular form of the baseline intensity of the Hawkes process $C_t$, $\lambda^\infty(i,\rho)$. By using Assumption (\ref{assump:baseline}), we have that
\[
\lambda^\infty(i,\rho)=\lambda_0+i\log(\alpha)+\rho\log(\beta)
\]
and substituting this into the PDDE (\ref{pde:f:2}) we have

\begin{equation}\label{pde:f:3}
\begin{aligned}
\dfrac{\partial }{\partial t}\psi(t,i,\rho,s)&+(rs-1) \frac{\partial}{\partial s}\psi(t,i,\rho,s)+rs\lambda^\infty(i,\rho)\psi(t,i,\rho,s)+\left(\frac{1}{\alpha}\right)^sM_Q(s)\frac{\partial}{\partial s}\psi(t,i-1,\rho,s)=\\[.3cm]
&(i+1)\mu \left(\alpha\beta\right)^s \psi(t,i+1,\rho-1,s)-i\mu\psi(t,i,\rho,s)
\end{aligned}
\end{equation}

By doing some algebra and using the fact $\psi(t,i,\rho,s)=0$ if $i<0$ or $\rho<0$, we have the following relations
\begin{align} \label{relation:series:1}
\sum\limits_{i=0}^\infty \sum\limits_{\rho=0}^\infty iz^iw^\rho\psi(t,i,\rho,s)&=z\frac{\partial}{\partial z}\varphi(t,z,w,s)\\[.3cm]
\sum\limits_{i=0}^\infty \sum\limits_{\rho=0}^\infty \rho z^iw^\rho\psi(t,i,\rho,s)&=w\frac{\partial}{\partial w}\varphi(t,z,w,s)\\[.3cm]
\sum\limits_{i=0}^\infty \sum\limits_{\rho=0}^\infty z^iw^\rho\psi(t,i-1,\rho,s)&=z\varphi(t,z,w,s)\\[.3cm] \label{relation:series:2}
\sum\limits_{i=0}^\infty \sum\limits_{\rho=0}^\infty (i+1)z^iw^\rho\psi(t,i+1,\rho-1,s)&=w\frac{\partial}{\partial z}\varphi(t,z,w,s)
\end{align}

Multiplying both sides by $z^iw^\rho$; adding those terms as a series; using Assumption (\ref{assump:baseline}); and using relation (\ref{defn:xi}), we obtain the result.

\end{proof}

\begin{proof}[Proof of Lemma \ref{lemma:firtsmom}]

Taking equation (\ref{pde:Gen:Fun}), substituting the value function $\varphi(t,z,w,s)$ and computing some of the derivatives within,
\begin{equation}\label{pde:Gen:Fun2}
\begin{aligned}
\dfrac{\partial }{\partial t}\Ex\left[z^{I_t}w^{R_t}e^{-s\Lambda(t)}\right]&+\left[rs-1+\left(\frac{1}{\alpha}\right)^sM_Q(s)z\right] \Ex\left[-\Lambda(t)z^{I_t}w^{R_t}e^{-s\Lambda(t)}\right]\\
&+\bigg[\left(\mu+rs\log\left(\alpha\right)\right)z-\mu\left(\alpha\beta\right)^sw\bigg]\Ex\left[I_t z^{I_t-1}w^{R_t}e^{-s\Lambda(t)}\right]\\[.3cm]
&-\bigg[rs\log(\beta)w\bigg]\Ex\left[R_t z^{I_t}w^{R_t-1}e^{-s\Lambda(t)}\right]=
-rs\lambda_0
\Ex\left[z^{I_t}w^{R_t}e^{-s\Lambda(t)}\right]
\end{aligned}
\end{equation}
Taking partial derivatives with respect to $s$ from equation (\ref{pde:Gen:Fun2}) and plugging in $(z,w,s)=(1,1,0)$,
\begin{equation}\label{pde:firstmom:s}
-\dfrac{\partial }{\partial t}\Ex\left[\Lambda(t)\right]-\bigg[r-\log\left(\alpha\right) -\Ex[Q]\bigg] \Ex\left[\Lambda(t)\right]+\bigg[r\log\left(\alpha\right)-\mu\log\left(\alpha\beta\right)\bigg]\Ex[I_t]-\bigg[r\log(\beta)\bigg]\Ex\left[R_t\right]=-r\lambda_0
\end{equation}
where we have used that $M_Q(0)=1$ and $\left.\frac{d}{ds}M_Q(s)\right|_{s=0}=-\Ex[Q]$. Similarly, taking partial derivatives with respect to $z$ from equation (\ref{pde:Gen:Fun2}) and plugging in $(z,w,s)=(1,1,0)$,
\begin{equation}\label{pde:firstmom:z}
\dfrac{\partial}{\partial t}\Ex\left[I_t\right]-\Ex\left[\Lambda(t)\right]+\mu\Ex\left[I_t\right]=0
\end{equation}
Finally, taking partial derivatives with respect to $w$ from equation (\ref{pde:Gen:Fun2}) and plugging in $(z,w,s)=(1,1,0)$,

\begin{equation}\label{pde:firstmom:w}
\dfrac{\partial }{\partial t}\Ex\left[R_t\right]-\mu\Ex\left[I_t\right]=0
\end{equation}
Arranging equations (\ref{pde:firstmom:s})-(\ref{pde:firstmom:w}) in matrix form, we get the result.
\end{proof}

%
%
%
%
%
%

\section{Simulation Algorithms}\label{Sec:algorithms}

\SetCustomAlgoRuledWidth{7in} 

When the baseline intensity $\lambda^\infty(i,\rho)$ is constant, several simulation algorithms are available. However, it is rare to find simulation algorithms for more general intensities. 

However, by using the classical thinning algorithm by Ogata, we can simulate our process, but some care and considerations need to be taken.

Before showing the main algorithm, we need to be able to evaluate the intensity at any time $t$ given the history. In this case, the history is provided in an array {\tt tIRarray} which is a $3\times n$ array with Row 1 having all the times at which $I_t$ or $R_t$ changed. Rows 2 and 3 have the value of $I_t$ and $R_t$, respectively, at the times on row 1.


\begin{algorithm}[H]
\SetAlgoLined
\KwResult{Compute $\Lambda(t)=\lambda^\infty(i,\rho)+\int_0^t Q_se^{-r(t-s)}dN_s$, for $t>0$\; }
\textcolor{Cerulean}{ Inputs:} $C_t$,$Q_t$, {\tt tIRarray},$t$,$\ldots$\;
\textcolor{Cerulean}{Output:} \textcolor{WildStrawberry}{{\tt Intensity}}=function()\;
\textcolor{Cerulean}{ Initialization}\;
 
 $n\;\leftarrow\; $length$(C_t)=||C_t||$\;
 $aux \;\leftarrow\; sum(C_t<t)=\int_0^t dC_s$\;
\textcolor{Cerulean}{Computation of the intensity function}\;
 \eIf{$(n=0\ ||\ aux=0)$}{
 \textbf{return} $\lambda^\infty(i_0,\rho_0)$\;
 }{
 	$Qs \;\leftarrow\; Q_t[1:aux]$\;
	$Cs \;\leftarrow\; C_t[1:aux]$\;
	$exps \;\leftarrow\;\exp(-r(t-Cs))$\;
	$integ \;\leftarrow\; sum(Qs*exps)$\;
	$lam \;\leftarrow\; \lambda^\infty(I_t,R_t)$\;
	\textbf{return} $lam+integ$\;
 }
 \caption{Computation of the Intensity function at time $t$.}
\end{algorithm}

Next, we will show the algorithm to simulate the state dependent Hawkes process. Notice how the intensity function needs to be updated between new infections as recoveries affect the intensity function also.

\begin{algorithm}[H]
\SetAlgoLined
\KwResult{Simulation of State-Dependent Hawkes Process with
intensity function $\Lambda(t)=\lambda^\infty(i,\rho)+\int_0^t Q_se^{-r(t-s)}dN_s$ }
\textcolor{Cerulean}{Inputs: } $T$, Qvals, Qprobs, $i_0,\rho_0,\ldots$\;
\textcolor{Cerulean}{Output:} \textcolor{WildStrawberry}{{\tt SDHP}}=function()\;
\textcolor{Cerulean}{Initialization}\;
 {\tt tIRarray} $\;\leftarrow\;[0,i_0,\rho_0]^T$\; 
 $t\;\leftarrow\;0$\;
 \While{$t \leq T$}{
 	$M \;\leftarrow\; $\textcolor{WildStrawberry}{{\tt Intensity}} $(t,\ldots)$\;
	\If{$M\leq0$}{
		\textbf{break}\;
	}
	$E\sim\text{Exp}(M)$\;
	$t \;\leftarrow\; t+E$\;
	$U\sim\text{Unif}(0,M)$\;
 }
 \textbf{\#\#} Since $\lambda^\infty(i,\rho)$ can change between increments of $C_t$ (via only a recovery), we need to update our intensity function between 0 and the next time candidate for a jump of $C_t$ which is $E$\;
 $\text{Raux} \;\leftarrow\; 0$\;
 \While{$\text{Raux}<E$}{
 	$\text{PossibleR}\sim\text{Exp}(\;\mu\cdot\text{Icurrent}\;)$\;
 	\If{$(\text{PossibleR+Raux})>E$}{
 		\textbf{break}
 	}
 	$\text{Raux} \;\leftarrow\;\text{Raux}+\text{PossibleR}$\;
 	\textbf{Update} {\tt tIRarray}\;
 }
 \If{$t<E$}{ 
 	\textbf{\#\#} Perform a thinning routine to accept new Infections\;
 	$\Lambda_N \;\leftarrow\; \textcolor{WildStrawberry}{{\tt Intensity}}(t,\ldots)$\;
 	\If{$U<\Lambda_N$}{
 			\textbf{Accept }$t$ as time of Infection\;
			\textbf{Update }all parameters\;
 		}
 	}
 \caption{Simulation of a trajectory of the State-dependent Hawkes SIR model}
 \label{algo:SIR}
\end{algorithm}

\newpage

\bibliographystyle{apalike}
\bibliography{submission_arxiv}

\end{document}